\documentclass{article}

     \usepackage[nonatbib, final]{neurips_2020}
     
\usepackage[utf8]{inputenc} 
\usepackage[T1]{fontenc}    
\usepackage{hyperref}       
\usepackage{url}            
\usepackage{booktabs}       
\usepackage{amsfonts}       
\usepackage{nicefrac}       
\usepackage{microtype}      

\usepackage{bbm}
\usepackage{amssymb}
\usepackage{amsthm}
\usepackage{mathtools}
\usepackage{graphicx}
\usepackage{epstopdf}
\DeclareGraphicsExtensions{.eps}
\usepackage{algorithm, amsmath}
\usepackage{algpseudocode}
\usepackage{float}
\usepackage{subcaption}
\usepackage{etoolbox}
\usepackage[T1]{fontenc}
\usepackage{multirow}
\usepackage{enumitem}
\usepackage{array}
\usepackage[numbers,sort&compress]{natbib}
\usepackage{diagbox}
\usepackage{authblk}

\DeclareMathOperator*{\argmax}{arg\,max}

\newcommand{\V}{\mathcal{V}}

\newcommand{\U}{\mathcal{U}}

\newcommand{\Z}{\mathcal{Z}}

\newcommand{\HH}{\mathcal{H}}

\newtheorem{theorem}{Theorem}
\newtheorem{lemma}{Lemma}
\newtheorem{corollary}{Corollary}

\newtheorem{assumption}{Assumption}

\newtheorem*{remark}{Remark}

\usepackage{tikz}
\usetikzlibrary{arrows.meta}
\tikzset{>={Latex[width=2mm,length=2mm]}}
\tikzstyle{vertex}=[circle, draw]

\tikzstyle{stuff_fill}=[vertex, style=green, fill=black!10]
\usetikzlibrary{calc}

\title{Robust Sequence Submodular Maximization}
\author[1]{Gamal Sallam}
\author[2]{Zizhan Zheng}
\author[1]{Jie Wu}
\author[1, 3]{Bo Ji }
\affil[1]{Department of Computer and Information Sciences, Temple University}
\affil[2]{Department of Computer Science, Tulane University }
\affil[3]{Department of Computer Science, Virginia Tech}

\begin{document}

\maketitle

\begin{abstract}
    Submodularity is an important property of set functions and has been extensively studied in the literature. It models set functions that exhibit a diminishing returns property, where the marginal value of adding an element to a set decreases as the set expands. This notion has been generalized to considering sequence functions, where the order of adding elements plays a crucial role and determines the function value; the generalized notion is called sequence (or string) submodularity. In this paper, we study a new problem of robust sequence submodular maximization with cardinality constraints. The robustness is against the removal of a subset of elements in the selected sequence (e.g., due to malfunctions or adversarial attacks). 
    Compared to robust submodular maximization for set function, new challenges arise when sequence functions are concerned. 
    Specifically, there are multiple definitions of submodularity for sequence functions, which exhibit subtle yet critical differences. Another challenge comes from two directions of monotonicity: forward monotonicity and backward monotonicity, both of which are important to proving performance guarantees. 
    To address these unique challenges, we design two robust greedy algorithms: while one algorithm achieves a constant approximation ratio but is robust only against the removal of a subset of contiguous elements, the other is robust against the removal of an arbitrary subset of the selected elements but requires a stronger assumption and achieves an approximation ratio that depends on the number of the removed elements.
    Finally, we generalize the analyses to considering sequence functions under weaker assumptions based on approximate versions of sequence submodularity and backward monotonicity. 
    
\end{abstract}

\section{Introduction}

Submodularity is an important property of set functions and has been extensively studied in the literature \cite{nemhauser1978best, khuller1999budgeted, calinescu2011maximizing, chekuri2014submodular}. It models set functions that exhibit a diminishing returns property, where the marginal value of adding an element to a set decreases as the set expands (i.e., contains more elements). The notion of submodularity has been generalized to considering sequence functions, where the order of adding elements plays a crucial role and determines the function value; the generalized notion is called sequence (or string) submodularity~\cite{zhang2016string, tschiatschek2017selecting, streeter2009online, zhang2012submodularity, alaei2010maximizing, jawaid2015submodularity, zhang2013near}. Several real-world applications, including machine learning based recommendation systems, ads allocation, and automation and control, involve the selection of elements in sequence. 
In this paper, we study a new problem of robust sequence submodular maximization with cardinality constraints. The robustness is against the removal of a subset of elements in the selected sequence (e.g., due to malfunctions or adversarial attacks).
To motivate the new problem studied in this paper, we begin with the discussions about two concrete applications (sensor activation and movie recommendation) and use them to illustrate the key differences between set functions and sequence functions. 

\textbf{Sensor Activation}. Consider the problem of sensor activation for moving target detection \cite{zhang2016string}, where the objective is to sequentially activate a certain number of sensors to maximize the probability of detecting a moving target. Suppose that each sensor covers a certain area. Without any prior knowledge of the target location or the probability of detection at each location, it is plausible to maximize the total area covered by the activated sensors. If the coverage area of each sensor remains constant over time, then it is sufficient to decide which subset of sensors to activate without concerning the order in which the sensors are activated. This scenario can typically be modeled as maximizing a (set) submodular function.

In practice, however, the coverage area of each sensor may decay over time for several reasons, such as battery decay and corrosive environment. Accounting for such factors, the coverage area of a sensor may be modeled as a decreasing function, e.g., in the form of $C e^{-t/{T}}$, where $C$ is the initial coverage area, $T$ is the sensor's lifetime, and $t = 0, 1, \dots$ is the time index. In this scenario, the sequence in which the sensors are activated is of critical importance. The activation sequence determines the total coverage area and thus impacts the probability of successful detection.

The above problem becomes more challenging if some sensors may be malfunctioning after they are activated or, even worse, if there is an adversary that may attack some sensors and render them non-working. Depending on how critical the application scenario is, one must ensure resilience of the sensor activation plan such that a certain performance (i.e., probability of successful detection) can still be guaranteed even in the worst-case failure scenario.

\textbf{Movie Recommendation}. Consider that a service provider (e.g., Netflix) would like to recommend movies to a user \cite{tschiatschek2017selecting, mitrovic2019adaptive}. 
It is common that the service provider recommends movies of similar flavor to a user. To some user, however, the incremental level of entertainment of watching a movie may decrease if the user had watched more similar movies, which exhibits a diminishing returns property.

In reality, however, the order in which the movies are recommended to the user may also impact how the user perceives a specific movie. In fact, movie recommendation and TV show recommendation have been modeled as sequence functions in \cite{tschiatschek2017selecting} and \cite{mitrovic2019adaptive}, respectively. As noted in the motivating example in \cite{mitrovic2019adaptive}, if the model determines that the user might be interested in \emph{The Lord of the Rings} series, then recommending \emph{The Return of the King} first and \emph{The
Fellowship of the Ring} last could make the user unsatisfied with an otherwise excellent recommendation. Moreover, the user may not watch all the recommended videos possibly because the user has already watched some of them or does not like them (e.g., due to low ratings and/or unfavorable reviews). 

\begin{table}[t]
  \caption{Representative work on submodular maximization \vspace{3pt}}
\label{table:related_work}
  \centering 
  \begin{tabular}{|c|c|c|}
  \cline{2-3}
\multicolumn{1}{c|}{} & (Set) Submodular Maximization & Sequence Submodular Maximization \\ \hline 
Non-robust& {\cite{nemhauser1978best, khuller1999budgeted, calinescu2011maximizing, chekuri2014submodular}} & {\cite{zhang2016string, tschiatschek2017selecting, streeter2009online, zhang2012submodularity, alaei2010maximizing, jawaid2015submodularity, zhang2013near}} \\ \hline 
Robust& {\raggedright \cite{orlin2018robust, bogunovic2017robust, mitrovic2017streaming, tzoumas2017resilient, iyer2019unified, tzoumas2018resilient, tzoumas2019robust, mirzasoleiman2017deletion, kazemi2018scalable, krause2008robust, powers2016constrained, powers2017constrained, anari2019structured}
} & {\raggedright \textbf{This paper}}\\ \hline
\end{tabular}
\end{table}

While the problem of submodular maximization is generally NP-hard due to its combinatorial nature, the property of submodularity has been exploited to design efficient approximation algorithms. Since the seminal work in \cite{nemhauser1978best}, it has been well known that a simple greedy algorithm and its variants can achieve an optimal approximation ratio\footnote{The approximation ratio is defined as the ratio of the objective value achieved by an algorithm over that achieved by an optimal algorithm.} of $(1-1/e)$ in various settings. Specifically, a variant of the greedy algorithm has also been shown to achieve the same approximation ratio for the problem of sequence submodular maximization \cite{streeter2009online}. Recently, robust versions of the submodular maximization problem have aroused a lot of research interests (e.g., \cite{orlin2018robust, tzoumas2017resilient, mitrovic2017streaming}). The focus of these studies is on selecting a set of elements that is robust against the removal of a subset of them. 


In this paper, we take one step further and consider a robust version of the sequence submodular maximization problem. The goal is to select a sequence of elements (i.e., elements in a specific order) with cardinality constraints such that the value of the sequence function is maximized when a certain number of the selected elements may be removed. 
In Table~\ref{table:related_work}, we position our work along with the literature on submodular maximization. The generalization to robust sequence submodular maximization introduces new challenges. As far as sequence functions are concerned, not only do the notions of submodularity and monotonicity involve variants of subtle yet critical differences, but the design and analysis of robust algorithms are also faced with novel technical difficulties, which render the proofs more challenging. In the sequel, we elaborate on these unique challenges.

First, there are two definitions of submodularity for set functions: (i) the marginal value of adding an element to a set decreases as the set expands; (ii) the marginal value of adding a set to another set decreases as the latter set expands. It is trivial to show that these two definitions are equivalent.
Replacing ``set'' with ``sequence'' in the above definitions gives two similar definitions for sequence functions. We will show that while (ii) still trivially implies (i) for sequence functions, the opposite does not hold in general. 
This leads to the following important question: Which definition of submodularity should one consider for sequence functions? Interestingly, while the weaker form (i) is sufficient for establishing provable approximation guarantees for non-robust sequence submodular maximization, one needs the stronger form (ii) to obtain similar results for the robust counterpart.

Second, while monotonicity is a straightforward notion for set functions (i.e., the value of a set function does not decrease as the set expands), there are two directions of monotonicity for sequence functions: forward monotonicity and backward monotonicity. Forward (resp., backward) monotonicity means that adding a sequence to the end (resp., beginning) of another sequence does not decrease the overall value. Both monotonicity properties are important to proving performance guarantees. 



Third, the impact of removing an element from a sequence depends both on the element itself and on its position in the sequence. This makes the robust algorithms designed for set functions inapplicable here and calls for new robust algorithms that are better suited for sequence functions. Besides, one needs stronger assumptions to proceed with performance analysis for sequence functions. Therefore, it is more important to prove performance guarantees under weaker assumptions based on approximate versions of submodularity and monotonicity, which are more likely to hold in practice.


Due to these unique challenges, it is unclear what conditions are sufficient for establishing provable approximation ratios for robust sequence submodular maximization, how to design efficient and robust algorithms, and how to prove performance guarantees for the designed algorithms. We aim to answer these questions in this paper. Our contributions are summarized as follows.
\begin{itemize}
    \item To the best of our knowledge, this is the first work that considers the problem of robust sequence submodular maximization. It is well known that the traditional (set) submodular maximization problem is already NP-hard. Accounting for sequence functions and robustness guarantees adds extra layers of difficulty, as the submodular and monotone properties of sequence functions involve variants with subtle yet critical differences.

    \item To address these unique challenges, we design two robust greedy algorithms for maximizing a forward-monotone, backward-monotone, and sequence-submodular function with cardinality constraints. While one algorithm achieves a constant approximation ratio but is robust only against the removal of a subset of contiguous elements, the other is robust against the removal of an arbitrary subset of the selected elements but requires a stronger assumption and achieves an approximation ratio that depends on the number of the removed elements. Although our proposed greedy algorithms are quite intuitive, the theoretical analysis is more challenging, and the presented approximation guarantees are highly nontrivial.
    
    \item We consider different definitions of submodularity and monotonicity and investigate their impacts on the derived theoretical results. 
    Our study reveals that compared to set functions, one needs more properties of sequence functions to establish similar approximation results.
    On the other hand, we introduce general versions of such properties, such as approximate sequence submodularity and approximate backward monotonicity, and leverage them to prove approximation results of the proposed algorithms under weaker assumptions, which are more likely to hold in practice. 
    We hope that this work serves as an important first step towards the
    design and analysis of efficient algorithms for robust sequence submodular maximization, which is worth further investigation through empirical evaluations for specific applications.

\end{itemize}


Due to space limitations, we provide all the proofs in the supplementary document.

\section{System Model and Problem Formulation}
\label{sec:system_model}
Consider a set of elements $\V$, with $V = |\V|$, where $|\cdot|$ denotes the cardinality of a set. Let $(v_1, \dots, v_m)$ be a sequence of non-repeated\footnote{This definition can be easily generalized to allow repetition by augmenting the ground set $\V$ as follows. Assume that each element $v_i \in \V$ can be repeated $z_i$ times. Let $v_i^j$ denote the $j$-th copy of element $v_i$. We use $\bar{\V}$ to denote the augmented ground set, which is defined as $\bar{\V} \triangleq \cup_{v_i \in \V} \{v_i^1, \dots, v_i^{z_i}\}$. Therefore, we can replace $\V$ with the augmented ground set $\bar{\V}$, which essentially allows the repetition of elements in $\V$.} elements selected over $m$ steps, where $v_i \in \V$ for $i = 1, \dots, m$, $v_i \neq v_j$ for $i \neq j$, and $m = 0, \dots, |\V|$. When $m = 0$, the sequence is empty and is denoted by $()$. We use $\HH(\V)$ to denote the set of all possible sequences of non-repeated elements in $\V$, and we use $\V(S)$ to denote the set of elements in sequence $S \in \HH(\V)$. By slightly abusing the notation, we use $|S|$ to denote the number of elements in sequence $S$, i.e., $|S|=|\V(S)|$.
Consider a sequence $S \in  \HH(\V)$ and a set $\U \subseteq \V$. We use $S - \U$ to denote the sequence that is constructed by removing all the elements in $\U$ from sequence $S$ without changing the order of the remaining elements. For instance, suppose $S = (v_2, v_1, v_5, v_3)$ and $\U = \{v_2, v_4, v_5\}$. Then, we have $S - \U = (v_1, v_3)$.
For two sequences $S_1, S_2 \in  \HH(\V)$, sequence $S_1$ is said to be a subsequence of sequence $S_2$ if we can write $S_1$ as $S_2 - \U$ for some $\U \subseteq \V(S_2)$.
Consider two sequences $S_1 = (v_1, \dots, v_{m_1})$ and $S_2 = (u_1, \dots, u_{m_2})$ in $\HH(\V)$, and let $S_2 - \V(S_1) = (w_1, \dots, w_{m_3})$. We define a concatenation of $S_1$ and $S_2$ as
    \begin{equation}
        \label{eq:concatenation_unique_elements}
        S_1 \oplus S_2  \triangleq (v_1, \dots, v_{m_1}, w_1, \dots, w_{m_3}).
    \end{equation}
Note that the concatenated sequence $S_1 \oplus S_2$ has no repeated elements.
We write $S_1 \preceq S_2$ if we can write $S_2$ as $S_1 \oplus S_3$ for some $S_3 \in \HH(\V)$. 

Before we define the problem of sequence submodular maximization, which was first considered in \cite{streeter2009online}, we introduce some important definitions. Consider a sequence function $h: \HH(\V) \rightarrow \mathbb{R}^+$, where $\mathbb{R}^+$ is the set of non-negative real numbers. Without loss of generality, we assume that the value of an empty sequence is zero, i.e., $h(())=0$. We define the marginal value of appending sequence $S_2$ to sequence $S_1$ as $h(S_2|S_1) \triangleq h(S_1 \oplus S_2) - h(S_1)$.
Function $h$ is said to be sequence-submodular if for all $S_3 \in \HH(\V)$, we have
        \begin{equation}
        \label{eq:sequence_submodular}
             h(S_3|S_1) \geq h(S_3|S_2), ~\forall S_1, S_2 \in \HH(\V)~ \text{ such that }  S_1 \preceq S_2.
        \end{equation}
The above inequality represents a diminishing returns property. Similarly, function $h$ is said to be element-sequence-submodular if for all $v \in \V$, we have
     \begin{equation}
        \label{eq:element_sequence_submodular}
         h((v)|S_1) \geq h((v)|S_2), ~\forall S_1, S_2 \in \HH(\V)~ \text{ such that }  S_1 \preceq S_2.
        \end{equation}

    From the above definitons, it is easy to see that a sequence-submodular function must also be element-sequence-submodular. However, an element-sequence-submodular function may not necessarily be sequence-submodular; we provide such an example in our supplementary material. This is in contrary to submodular (set) functions, for which one can easily verify that similar definitions of Eqs. \eqref{eq:sequence_submodular} and \eqref{eq:element_sequence_submodular} imply each other. Although it is noted (without a proof) in \cite{zhang2016string} that using an induction argument, one can show that an element-sequence-submodular function must also be sequence-submodular, we find this claim false due to the counterexample we find (see our supplementary material). 
    
    Also, function $h$ is said to be forward-monotone if
        \begin{equation}
            \label{eq:forward_monotone}
            h(S_1 \oplus S_2) \geq h(S_1),~ \forall S_1, S_2 \in \HH(\V), 
        \end{equation}    
   and is said to be backward-monotone if
        \begin{equation}
         \label{eq:backward_monotone}
            h(S_1 \oplus S_2) \geq h(S_2),~ \forall S_1, S_2 \in \HH(\V).
        \end{equation}
For the sensor activation example we discussed in the introduction, forward monotonicity (resp., backward monotonicity) means that adding a sequence of sensors to the end (resp., the beginning) of another sequence of sensors does not reduce the total coverage area. We will later introduce approximate versions of sequence submodularity and backward monotonicity and generalize the theoretical results under weaker assumptions based on such generalized properties (see Section~\ref{sec:approx}).

The problem of selecting a sequence $S \in \HH(\V)$ with an objective of maximizing function $h$ with cardinality constraints (i.e., selecting no more than $k$ elements for $k>0$) can be formulated as
\begin{equation}\tag{$P$}
    \label{eq:sequence_selection}
    \begin{aligned}
    & \underset{S \in \HH(\V),~ |S| \leq k}{\max} \quad   h(S).
    \end{aligned}
    \end{equation}

Next, we propose a robust version of Problem \eqref{eq:sequence_selection}, which accounts for the removal of some of the selected elements. Consider $\tau \leq k$. The robust version of Problem \eqref{eq:sequence_selection} can be formulated as
\begin{equation}\tag{$R$}
    \label{eq:Robust_sequence_selection}
    \begin{aligned}
    \underset{S \in \HH(\V),~ |S| \leq k~}{\max} \underset{\V^\prime  \subseteq \V(S),~ |\V^\prime| \leq \tau}{\min} \quad   h(S - \V^\prime).
    \end{aligned}
    \end{equation}

Without loss of generality, we assume $k>1$ for Problem \eqref{eq:Robust_sequence_selection}. In the next section, we discuss the challenges of Problem \eqref{eq:Robust_sequence_selection} and present the proposed robust algorithms. 

\section{Proposed Robust Algorithms}
\label{sec:proposed}

We begin with a discussion about the non-robust sequence submodular maximization problem (Problem~\eqref{eq:sequence_selection}), through which we provide useful insights into the understanding of the challenges of Problem~\eqref{eq:Robust_sequence_selection}. Although Problem \eqref{eq:sequence_selection} is NP-hard, it can be approximately solved using a simple \emph{Sequence Submodular Greedy} (SSG) algorithm \cite{streeter2009online}. Under the SSG algorithm, we begin with an empty sequence $S$; in each iteration, we choose an element that leads to the largest marginal value with respect to $S$ and append it to sequence $S$, i.e., $S = S \oplus \argmax_{v \in \V \setminus \V(S)} h((v)| S)$. We repeat the above procedure until $k$ elements have been selected. It has been shown in \cite{streeter2009online} that the SSG algorithm achieves an approximation ratio of $(1 - 1/e)$ for maximizing a forward-monotone, backward-monotone, and element-sequence-submodular function with cardinality constraints. 

Although the SSG algorithm approximately solves Problem~\eqref{eq:sequence_selection}, it can perform very poorly if one directly applies it to solving its robust counterpart (Problem~\eqref{eq:Robust_sequence_selection}). The intuition behind this is the following. The SSG algorithm tends to concentrate the value of the selected sequence on the first few elements. Selecting elements in this manner leaves the overall sequence vulnerable as removing some of these elements would have a high impact on the overall value. Consider the following example, where we assume $\tau = 1$ for simplicity. Let $\V_1 = \{v\}$, $ \V_2 = \{u_1, \dots, u_n\}$, $ \V_3 = \{w_1, \dots, w_n\}$, and $\V = \V_1 \cup \V_2 \cup \V_3$. Assume $h((v)) = 1$, $h((u_i)) = 1/n$ for all $u_i \in \V_2$, and $h((w_i)) = \epsilon$ for all $w_i \in \V_3$, where $\epsilon$ is an arbitrarily small positive number. Also, assume that for any $S_1, S_2 \in \HH(\V)$ such that $v \in \V(S_1)$ and $v \notin \V(S_2)$, we have $h((u_i) | S_1) = 0$ and $h((u_i) | S_2) = 1/n$ for all $u_i \in \V_2$ and $h((w_i) | S_1) = h((w_i) | S_2) = \epsilon$ for all $w_i \in \V_3$. Suppose $k = n$. Then, the SSG algorithm will select $v$ as the first element and select the subsequent $n-1$ elements from $\V_3$. The value of the selected sequence will be $1 + (n-1)\epsilon$. If element $v$ is removed, then the value of the remaining sequence will be $(n-1)\epsilon$, which can be arbitrarily small. In contrast, a sequence consisting of $n$ elements from $\V_2$ will be robust against the removal of any element. This is because the overall value is equally distributed across all the elements in the sequence and the value of the sequence after removing any element is $(n-1)/n$. 

The above example shows that the SSG algorithm may perform arbitrarily bad for Problem~\eqref{eq:Robust_sequence_selection}. To that end, we propose two greedy algorithms that can address this limitation and ensure robustness of the selected sequence for Problem~\eqref{eq:Robust_sequence_selection}. First, we propose an algorithm that achieves a constant approximation ratio but is robust only against the removal of $\tau$ contiguous elements (Section~\ref{sec:contiguous}). Then, we further propose an algorithm that works in a general setting without the contiguous restriction and is robust against the removal of an arbitrary subset of $\tau$ selected elements, but it requires a stronger assumption and achieves an approximation ratio that depends on the number of removed elements (Section~\ref{sec:general}).

\subsection{Robustness Against the Removal of Contiguous Elements} \label{sec:contiguous}

In this subsection, we wish to design an algorithm that is robust against the removal of $\tau$ contiguous elements. The assumption of the removal of contiguous elements can model a spatial relationship such as sensors in close proximity or a temporal relationship such as consecutive episodes of a TV show. We design a variant of the SSG algorithm that approximately solves Problem~\eqref{eq:Robust_sequence_selection}. The algorithm is presented in Algorithm~\ref{alg:robustGreedy_tau_contiguous}. As we discussed earlier, the limitation of the SSG algorithm is that the selected sequence is vulnerable because the overall value might be concentrated in the first few elements. Algorithm~\ref{alg:robustGreedy_tau_contiguous} is motivated by this key observation and works in two steps. In Step 1, we select a sequence $S_1$ of $\tau$ elements from elements in $\V$ in a greedy manner as in SSG. In Step 2, we select another sequence $S_2$ of $k - \tau$ elements from elements in $\V \setminus \V(S_1)$, again in a greedy manner as in SSG. Note that when we select sequence $S_2$, we perform the greedy selection as if sequence $S_1$ does not exist at all. This ensures that the value of the final sequence $S = S_1 \oplus S_2$ is not concentrated in either $S_1$ or $S_2$. The complexity of Algorithm~\ref{alg:robustGreedy_tau_contiguous} is $O(k V)$, which is in terms of the number of function evaluations used in the algorithm.


We first state the following assumption that is needed for deriving the main results in this subsection.
\begin{assumption}
\label{assump:forward_backward_sequence}
Function $h$ is forward-monotone, backward-monotone, and sequence-submodular.
\end{assumption}

In Theorem~\ref{theorem:robustGreedy_tau_1_case_a}, we state the approximation result of Algorithm~\ref{alg:robustGreedy_tau_contiguous} in a special case of $\tau =1$. We consider this special case for two reasons: (i) it is easier to explain the key ideas in the proof of this special case; (ii) we can prove better approximation ratios in this special case, which may not be obtained from the analysis in the case of $1 \le \tau \le k$.


\begin{theorem}
\label{theorem:robustGreedy_tau_1_case_a}
Consider $\tau =1$. Under Assumption~\ref{assump:forward_backward_sequence}, Algorithm~\ref{alg:robustGreedy_tau_contiguous} achieves an approximation ratio of $\max \left\{ \frac{e-1}{2e}, \frac{e^{\frac{k-2}{k-1}} - 1}{2 e^{\frac{k-2}{k-1}} - 1} \right\}$, which is lower bounded by a constant $\frac{e-1}{2e}$.
\end{theorem}


\begin{remark}
The two terms of the approximation ratio in Theorem \ref{theorem:robustGreedy_tau_1_case_a} have different advantages. While the first term remains constant (i.e., $\frac{e-1}{2e} \approx 0.316$), the second term (i.e., $(e^{\frac{k-2}{k-1}} - 1)/(2 e^{\frac{k-2}{k-1}} - 1)$) is a monotonically increasing function of $k$. The first term is larger for a small value of $k$ (when $k < 4$); the second term becomes larger for a wide range of $k$ (when $k \ge 4$). 
\end{remark}

In Theorem~\ref{theorem:robustGreedy_tau_contigous_case_a}, we state the approximation result of Algorithm~\ref{alg:robustGreedy_tau_contiguous} in the case of $1 \le \tau \le k$.

\begin{theorem}
\label{theorem:robustGreedy_tau_contigous_case_a}
Consider $1 \le \tau \le k$. Under Assumption~\ref{assump:forward_backward_sequence}, Algorithm~\ref{alg:robustGreedy_tau_contiguous} achieves an approximation ratio of $\max \left \{ \frac{(e - 1)^2}{e(2e-1)}, \frac{(e - 1)(e^{\frac{k-2\tau}{k-\tau}} -1)}{(2e - 1)e^{\frac{k-2\tau}{k-\tau}}- (e - 1)} \right \}$, which is lower bounded by a constant $\frac{(e - 1)^2}{e(2e-1)}$.
\end{theorem}


\begin{remark}
The two terms of the approximation ratio in Theorem~\ref{theorem:robustGreedy_tau_contigous_case_a} have different advantages. While the first term remains constant (i.e., $\frac{(e - 1)^2}{e(2e-1)} \approx 0.245$), the second term is a monotonically increasing (resp., decreasing) function of $k$ (resp., $\tau$). 
The first term is larger when $k < \frac{2 - \ln(\frac{e^2 + e - 1}{2e - 1})}{1 - \ln(\frac{e^2 + e - 1}{2e - 1})} \tau$; the second term becomes larger when $k \ge \frac{2 - \ln(\frac{e^2 + e - 1}{2e - 1})}{1 - \ln(\frac{e^2 + e - 1}{2e - 1})} \tau$. We provide the approximation ratio under different values of $\tau$ and $k$ in Table~\ref{tab:theorem_2_values} in the Appendix of the supplementary document.
\end{remark}

\subsection{Robustness Against the Removal of Arbitrary Elements} \label{sec:general}

\begin{figure}[t]
\begin{minipage}{0.50\textwidth}
\begin{algorithm}[H]
    \centering
     \caption{Robust greedy algorithm against the removal of contiguous elements}
        \label{alg:robustGreedy_tau_contiguous}
        \begin{algorithmic}[1]
            \State {Input: $\V, k, \tau$; Output: $S$}
            \State {Initialization: $S = S_1 = S_2 = ()$}
		    \item[] $//$\textbf{Step 1:}
            \While {$|S_1| < \tau$}
            \State {$S_1 = S_1 \oplus \argmax_{v \in \V \setminus \V(S_1)} h((v)| S_1)$}
            \EndWhile
            \item[] $//$\textbf{Step 2:}
            \While {$|S_2| < k-
            \tau$}
            \State {$S_2 = \newline S_2 \oplus \argmax_{v \in \V \setminus (\V(S_1) \cup \V(S_2))} h((v)| S_2)$}
            \EndWhile
            \State{$S = S_1 \oplus S_2$}
    \end{algorithmic}
\end{algorithm}
\end{minipage}
\hfill
\begin{minipage}{0.48\textwidth}
\begin{algorithm}[H]
    \centering
     \caption{Robust greedy algorithm against the removal of arbitrary elements}
        \label{alg:robustGreedy_tau_general}
    \begin{algorithmic}[1]
        \State {Input: $\V, k, \tau$; Output: $S$}
            \State {Initialization: $S = S_1 = S_2 = ()$}
            \item[] $//$\textbf{Step 1:}
            \While {$|S_1| < \tau$}
            \State {$S_1 = S_1 \oplus \argmax_{v \in \V \setminus \V(S_1)} h((v))$}
            \EndWhile
           \item[] $//$\textbf{Step 2:}
            \While {$|S_2| < k-\tau$}
            \State {$S_2 = \newline S_2 \oplus \argmax_{v \in \V \setminus (\V(S_1) \cup \V(S_2))} h((v)| S_2)$}
            \EndWhile
            \State{$S = S_1 \oplus S_2$}
    \end{algorithmic}
\end{algorithm}
\end{minipage}
\end{figure}
In this subsection, we wish to design an algorithm that is robust against the removal of an arbitrary subset of $\tau$ selected elements, which are not necessarily contiguous.    
One weakness of Algorithm \ref{alg:robustGreedy_tau_contiguous} is that the value of the selected sequence could be concentrated in the first few elements of subsequences $S_1$ and $S_2$. If we allow the removal of an arbitrary subset of $\tau$ selected elements, the removal of the first few elements of subsequences $S_1$ and $S_2$ could leave the remaining sequence with little or no value. By considering a special case in Section~\ref{sec:contiguous} where the removed elements are restricted to be contiguous, we have managed to prevent such worst case from happening. However, the problem becomes more challenging when we consider a more general case without such a restriction. In the following, we propose an algorithm that is robust against the removal of an arbitrary subset of $\tau$ selected elements, but it requires a stronger assumption and achieves an approximation ratio that depends on the value of $\tau$. This new algorithm is presented in Algorithm~\ref{alg:robustGreedy_tau_general}.

 Similar to Algorithm~\ref{alg:robustGreedy_tau_contiguous}, Algorithm~\ref{alg:robustGreedy_tau_general} works in two steps. However, there is a subtle yet critical difference in Step 1, which is the key to ensuring robustness in the general case. Specifically, in Step 1 of Algorithm~\ref{alg:robustGreedy_tau_general}, we select a sequence $S_1$ of $\tau$ elements from $\V$ by iteratively choosing an element $v$ in a greedy manner, based on its absolute value $h((v))$ instead of its marginal value $h((v)|S_1)$ as in Algorithm~\ref{alg:robustGreedy_tau_contiguous}. We then select a sequence $S_2$ in Step 2, which is the same as that of Algorithm~\ref{alg:robustGreedy_tau_contiguous}. The final output is $S = S_1 \oplus S_2$. Algorithm~\ref{alg:robustGreedy_tau_general} also has a complexity of $O(k V)$.


Before we state the approximation results of Algorithm~\ref{alg:robustGreedy_tau_general}, we introduce a generalized definition of sequence submodularity.
Function $h$ is said to be general-sequence-submodular if for all $S_3 \in \HH(\V)$, we have
    \begin{equation}
        \label{eq:general_sequence_submodular}
        h(S_3|S_1) \geq h(S_3|S_2), ~\forall S_1, S_2 \in \HH(\V)~ \text{such that $S_1$ is a subsequence of $S_2$}.
    \end{equation}
Note that $S_1 \preceq S_2$ implies that $S_1$ is a subsequence of $S_2$, but not vice versa. Therefore, the general sequence submodularity defined above generalizes the sequence submodularity defined in Eq.~\eqref{eq:sequence_submodular} as a special case with  $S_1 \preceq S_2$.
Next, we state Assumption~\ref{assump:forward_backward_general} and Theorem~\ref{theorem:robustGreedy_general_tau_case_a}.

\begin{assumption}
\label{assump:forward_backward_general}
Function $h$ is forward-monotone, backward-monotone, and general-sequence-submodular.
\end{assumption}

\begin{theorem}
\label{theorem:robustGreedy_general_tau_case_a}
Consider $1 \le \tau \le k$. Under Assumption~\ref{assump:forward_backward_general}, Algorithm \ref{alg:robustGreedy_tau_general} achieves an approximation ratio of $\frac{1- 1/e}{1 + \tau}$.
\end{theorem}
\begin{remark}
While we only need the simplest form of the diminishing returns definition (element-sequence-submodularity) to establish provable approximation guarantees for the non-robust sequence submodular maximization, for its robust counterpart, we require stronger assumptions (sequence-submodularity and general-sequence-submodularity vs. element-sequence-submodularity) to show provable performance guarantees. In addition, consider a set function $r$ and ground set $\V$. While monotonicity of set function $r$ implies monotonicity of the same function with respect to the marginal value of adding a set to another set (i.e., $r(\V_2|\V_1) \triangleq r(\V_1 \cup \V_2) - r(\V_1)$ for any $\V_1, \V_2 \subseteq \V$), a similar property does not hold for backward monotonicity of sequence functions. This subtle difference results in a more involved analysis of showing similar results for sequence functions.
\end{remark}

\section{Robust Approximate Sequence Submodular Maximization} \label{sec:approx}

In this section, we introduce generalized versions of sequence submodularity and backward monotonicity, which are called approximate sequence submodularity and approximate backward monotonicity. Then, we show that Algorithms \ref{alg:robustGreedy_tau_contiguous} and \ref{alg:robustGreedy_tau_general} can also approximately solve Problem~\eqref{eq:Robust_sequence_selection} under weaker assumptions based on such generalized properties. 

We begin with some additional definitions.
Consider $\mu_1 \in (0, 1]$. Function $h$ is said to be $\mu_1$-element-sequence-submodular if for all $v \in \V$, we have
        \begin{equation}
        \label{eq:element_sequence_submodular_approximate}
        h((v)|S_1) \geq \mu_1 h((v)|S_2), \forall S_1, S_2 \in \HH(\V)~ \text{ such that }  S_1 \preceq S_2.
        \end{equation}
Also, consider $\mu_2 \in (0, 1]$. Function $h$ is said to be $\mu_2$-sequence-submodular if for all $S_3 \in \HH(\V)$, we have
\begin{equation}
        \label{eq:sequence_submodular_approximate}
         h(S_3|S_1) \geq \mu_2 h(S_3|S_2), ~\forall S_1, S_2 \in \HH(\V)~ \text{ such that }  S_1 \preceq S_2.
\end{equation}
Note that $\mu_1$ could be greater than $\mu_2$ for some function $h$. We distinguish $\mu_1$ and $\mu_2$ as some of our results depend on $\mu_1$ only.
Similarly, consider $\mu_3 \in (0, 1]$. Function $h$ is said to be $\mu_3$-general-sequence-submodular if for all $S_3 \in \HH(\V)$, we have
\begin{equation}
        \label{eq:general_sequence_submodular_approximate}
       h(S_3|S_1) \geq \mu_3 h(S_3|S_2), ~\forall S_1, S_2 \in \HH(\V)~ \text{ such that $S_1$ is a subsequence of $S_2$}.
\end{equation}
Consider $\alpha \in (0, 1]$. Function $h$ is said to be $\alpha$-backward-monotone if
        \begin{equation}
          \label{eq:backward_monotone_approximate}
           h(S_1 \oplus S_2) \geq \alpha h(S_2),~ \forall S_1, S_2 \in \HH(\V).
        \end{equation}

Next, we state several assumptions that will be needed for deriving the main results in this section.

\begin{assumption}
\label{assump:forward_backward_mu1_mu2}
Function $h$ is forward-monotone, backward-monotone,  $\mu_1$-element-sequence-submodular, and $\mu_2$-sequence-submodular.
\end{assumption}

\begin{assumption}
\label{assump:forward_alpha_mu1_mu2}
Function $h$ is forward-monotone, $\alpha$-backward-monotone,  $\mu_1$-element-sequence-submodular, and $\mu_2$-sequence-submodular.
\end{assumption}

\begin{assumption}
\label{assump:forward_alpha_mu1_mu3}
Function $h$ is forward-monotone, $\alpha$-backward-monotone,  $\mu_1$-element-sequence-submodular, and $\mu_3$-general-sequence-submodular.
\end{assumption}



We are now ready to state the generalized approximation results of Algorithm~\ref{alg:robustGreedy_tau_contiguous} under Assumptions \ref{assump:forward_backward_mu1_mu2} and \ref{assump:forward_alpha_mu1_mu2}, respectively.

\begin{theorem}
\label{theorem:robustGreedy_tau_1_case_b}
Consider $\tau =1$. Under Assumption~\ref{assump:forward_backward_mu1_mu2}, Algorithm~\ref{alg:robustGreedy_tau_contiguous} achieves an approximation ratio of $\frac{a(e^{b} - 1)}{e^{b} - a}$, where $a = \frac{\mu_1 \mu_2}{\mu_1 + 1}$ and $b = \mu_1 \cdot \frac{k -2}{k-1}$; under Assumption~ \ref{assump:forward_alpha_mu1_mu2}, Algorithm~\ref{alg:robustGreedy_tau_contiguous} achieves an approximation ratio of $\frac{\alpha^2 \mu_1 \mu_2 (e^{\mu_1} - 1)}{(\mu_1  + \alpha  )e^{\mu_1}}$.
\end{theorem}

\begin{theorem}
\label{theorem:robustGreedy_tau_contigous_case_b}
Consider $1 \le \tau \le k$. Under Assumption~\ref{assump:forward_backward_mu1_mu2}, Algorithm~\ref{alg:robustGreedy_tau_contiguous} achieves an approximation ratio of $\frac{ a \mu_2 (e^b - 1)}{(a+1)e^b - a \mu_2}$, where $a = \mu_1 \cdot (1-1/e^{\mu_1 })$ and $b = \mu_1 \cdot \frac{k-2\tau}{k-\tau}$; under Assumption~ \ref{assump:forward_alpha_mu1_mu2}, Algorithm~\ref{alg:robustGreedy_tau_contiguous} achieves an approximation ratio of $\frac{\alpha^2 \mu_1 \mu_2 (e^{\mu_1} - 1)^2}{\mu_1 e^{\mu_1}(e^{\mu_1} - 1) + e^{2\mu_1}}$.
\end{theorem}




Finally, we state the approximation result of Algorithms~\ref{alg:robustGreedy_tau_general} under Assumption~\ref{assump:forward_alpha_mu1_mu3}.


\begin{theorem}
\label{theorem:robustGreedy_tau_general_case_b}
Consider $1 \le \tau \le k$. Under Assumption~\ref{assump:forward_alpha_mu1_mu3}, Algorithm~\ref{alg:robustGreedy_tau_general} achieves an approximation ratio of $\frac{\alpha^2 \mu_1 \mu_3 (e^{\mu_1} - 1)}{(\mu_1  + \alpha   \tau)e^{\mu_1}}$.
\end{theorem}

\section{Related Work}
Since the seminal work in \cite{nemhauser1978best}, submodular maximization has been extensively studied in the literature. Several efficient approximation algorithms have been developed for maximizing a submodular set function in various settings (e.g., \cite{nemhauser1978best, khuller1999budgeted, calinescu2011maximizing, chekuri2014submodular}). The concept of sequence (or string) submodularity for sequence functions is a generalization of submodularity, which has been introduced recently in several studies (e.g., \cite{zhang2016string, tschiatschek2017selecting, streeter2009online, zhang2012submodularity, alaei2010maximizing, jawaid2015submodularity, zhang2013near, mitrovic2018submodularity}). In \cite{streeter2009online}, it has been shown that a simple greedy algorithm can achieve an approximation ratio of $(1 - 1/e)$ for maximizing a forward-monotone, backward-monotone, and element-sequence-submodular function.

On the other hand, robust versions of submodular maximization has been considered in some recent studies (e.g., \cite{orlin2018robust, bogunovic2017robust, mitrovic2017streaming, tzoumas2017resilient}), where the focus is on selecting a set of elements that is robust against the removal of a subset of them. In \cite{orlin2018robust}, the authors propose the first algorithm with a constant approximation ratio for the problem of robust submodular maximization with cardinality constraints, where the selected set is of size $k$ and the robustness is against the removal of any $\tau$ elements of the selected set. The constant approximation ratio derived in \cite{orlin2018robust} is valid as long as the number of removed elements is small compared to the selected set (i.e., $\tau = o(\sqrt{k})$). An extension that guarantees the same constant approximation ratio but allows the removal of a larger number of elements (i.e., $\tau = o(k)$) is presented in \cite{bogunovic2017robust}. Another algorithm that allows the removal of an arbitrary number of elements under a mild assumption is presented in \cite{mitrovic2017streaming}. The work in \cite{tzoumas2017resilient} relaxes the restriction on $\tau$, but the achieved approximation ratio depends on the value of $\tau$. The work in \cite{iyer2019unified} considers the same problem under different types of constraints, such as matroid and knapsack constraints. The work in \cite{tzoumas2018resilient, tzoumas2019robust} extends the work in \cite{tzoumas2017resilient} to a multi-stage setting, where the decision at one stage takes into account the failures that happened in the previous stages. Other extensions that consider fairness and privacy issues are studied in \cite{mirzasoleiman2017deletion, kazemi2018scalable}. It is unclear whether all of these algorithms for robust set submodularity can be properly extended to our problem, as converting a set into a sequence could result in an arbitrarily bad performance. Even if so, it is more likely that establishing their approximation guarantees would require a more sophisticated analysis, which calls for more in-depth investigations. Note that the analysis of our simple greedy algorithms is already very sophisticated.

In \cite{krause2008robust}, a different notion of robustness is considered, which is referred to as maximizing the minimum of multiple submodular functions. This work proposes a bicriterion approximation algorithm for the studied problem with cardinality constraints. Moreover, the work in \cite{powers2016constrained, powers2017constrained, anari2019structured} extends that of \cite{krause2008robust} to accommodate a wide variety of constraints, including matroid and knapsack constraints. 
The work in \cite{bogunovic2018robust} develops an approximation algorithm for robust non-submodular maximization, using other characterizations such as the submodularity ratio and the inverse curvature. The work in \cite{mitrovic2019adaptive} introduces the idea of adaptive sequence submodular maximization, which aims to utilize the feedback obtained in previous iterations to improve the current decision. Note that while the work in \cite{tschiatschek2017selecting, mitrovic2019adaptive, mitrovic2018submodularity} assumes that the sequential relationship among elements is encoded as a directed acyclic graph, we consider a general setting without such structures. It would indeed be interesting to explore our algorithms when the sequential relationship is encoded in a specific graphical form.


\section{Conclusion}
In this paper, we investigated a new problem of robust sequence submodular maximization. We discussed the unique challenges introduced by considering sequence functions and ensuring robustness guarantees. 
To address these novel challenges, we proposed two robust greedy algorithms and proved that they can achieve certain approximation ratios for the considered problem, assuming forward-monotone, backward-monotone, and sequence-submodular functions. We further introduced approximate versions of sequence submodularity and backward monotonicity and showed that the proposed algorithms can also provide performance guarantees under a larger class of weaker assumptions based on such generalized properties. Our future work includes developing more efficient algorithms with better approximation ratios in the general settings and investigating the possibility of obtaining similar results under the assumption of generalized/approximate forward monotonicity.



\section{Broader Impact}
This work contributes to the state-of-the-art theory of submodular optimization. The proposed algorithms and the presented approximation results can be applied to real-world applications where the stated assumptions of sequence submodularity and monotonicity or their approximate versions are satisfied. Several real-world applications,  including machine learning based recommendation systems, ads allocation, and automation and control, involve the selection of elements in sequence.

\section*{Acknowledgement}
This work was supported in part by the NSF under Grants CNS-1651947, CNS-1824440, CNS-1828363, and CNS-1757533.

\bibliographystyle{IEEEtran}
\bibliography{references}

\newpage

\section{Appendix}
\subsection{Element-sequence-submodular vs. Sequence-submodular: A Counterexample}
\label{counterexample:general_sequence_submodular_defintion}

In \cite{zhang2016string}, it is noted (without a proof) that an element-sequence-submodular function must also be sequence-submodular. However, we find this claim false and present a counterexample in Table~\ref{tab:sequence_vs_element_submodular} to show that an element-sequence-submodular function is not necessarily sequence-submodular. In the presented example, it is easy to verify that function $h$ is element-sequence-submodular, i.e., Eq.~\eqref{eq:element_sequence_submodular} is satisfied; however, it is not sequence-submodular, i.e., Eq.~\eqref{eq:sequence_submodular} is not satisfied. Specifically, we have $h((v_2, v_3)|()) < h((v_2, v_3) | (v_1))$ due to
$h((v_2, v_3)|()) = h((v_2, v_3)) - h(()) = 1.2 - 0 = 1.2$ and $h((v_2, v_3) | (v_1)) = h((v_1, v_2, v_3)) - h((v_1)) = 2.2 - 0.2 = 2$.

\begin{table}[h]
    \caption{A counterexample showing that an element-sequence-submodular function is not necessarily sequence-submodular, where $\V = \{v_1, v_2, v_3\}$ and the columns are input sequences, values of function $h$, and marginal values of $v_1$, $v_2$, and $v_3$ with respect to sequence $S$, respectively \vspace{3pt}}
    \label{tab:sequence_vs_element_submodular}
    \centering
    \begin{tabular}{|c|c|c|c|c|} \hline 
\backslashbox{Sequence $S$}{Function $h$} & $h(S)$ & $h((v_1)|S)$ & $h((v_2)|S)$& $h((v_3)|S)$  \\ \hline
$()$ & 0 & 0.2 & 1.2 & 1  \\ \hline
$(v_1)$ & 0.2 & 0 & 1 & 1   \\ \hline
$(v_2)$ & 1.2 & 0 & 0 & 0 \\ \hline
$(v_3)$ & 1 & 0.2 & 1 & 0 \\ \hline
$(v_1, v_2)$ & 1.2 & 0 & 0& 1 \\ \hline
$(v_2, v_1)$ & 1.2 & 0 & 0& 0 \\ \hline
$(v_1, v_3)$ & 1.2 & 0 & 1 & 0 \\ \hline
$(v_3, v_1)$ & 1.2 & 0 & 1 & 0 \\ \hline
$(v_2, v_3)$ & 1.2 & 0 & 0 & 0 \\ \hline
$(v_3, v_2)$ & 2 & 0.2 & 0 & 0 \\ \hline
$(v_1, v_2, v_3)$ & 2.2 & 0 & 0 & 0 \\ \hline
$(v_1, v_3, v_2)$ & 2.2 & 0 & 0 & 0 \\ \hline
$(v_2, v_1, v_3)$ & 1.2 & 0 & 0 & 0 \\ \hline
$(v_2, v_3, v_1)$ & 1.2 & 0 & 0 & 0 \\ \hline
$(v_3, v_1, v_2)$ & 2.2 & 0 & 0 & 0 \\ \hline
$(v_3, v_2, v_1)$ & 2.2 & 0 & 0 & 0 \\ \hline
    \end{tabular}
\end{table}

\subsection{Approximation Ratio of Theorem \ref{theorem:robustGreedy_tau_contigous_case_a} with Different $\tau$ and $k$}
\label{approximation_ratio_theorem_2}

\begin{table}[h]
    \caption{Approximation ratio of Theorem~\ref{theorem:robustGreedy_tau_contigous_case_a} with different values of $\tau$ and $k$ \vspace{3pt}}
    \label{tab:theorem_2_values}
    \centering
    \begin{tabular}{|c|c|c|c|c||c|c|c|c|c|c|} \hline 
\backslashbox{$\tau$}{$k$} & 50 & 52 & 54 & 56  & 58 & 60 & 62 & 64 & 66 & 68 \\ \hline
2 & 0.28& 0.281& 0.281& 0.281& 0.281& 0.281& 0.281& 0.282& 0.282& 0.282  \\ \hline
4 & 0.275& 0.275& 0.275& 0.276& 0.276& 0.277& 0.277& 0.277& 0.277& 0.278   \\ \hline
6 & 0.268& 0.269& 0.269& 0.27& 0.271& 0.271& 0.272& 0.272& 0.273& 0.273 \\ \hline
8 & 0.26& 0.261& 0.262& 0.263& 0.264& 0.265& 0.266& 0.267& 0.267& 0.268 \\ \hline
10 & 0.25& 0.252& 0.254& 0.255& 0.257& 0.258& 0.259& 0.261& 0.262& 0.262 \\ \hline
12 & 0.245& 0.245& 0.245& 0.246& 0.248& 0.25& 0.252& 0.253& 0.255& 0.256 \\ \hline
14 & 0.245& 0.245& 0.245& 0.245& 0.245& 0.245& 0.245& 0.245& 0.247& 0.249 \\ \hline
16 & 0.245& 0.245& 0.245& 0.245& 0.245& 0.245& 0.245& 0.245& 0.245& 0.245 \\ \hline
18 & 0.245& 0.245& 0.245& 0.245& 0.245& 0.245& 0.245& 0.245& 0.245& 0.245 \\ \hline
20 & 0.245& 0.245& 0.245& 0.245& 0.245& 0.245& 0.245& 0.245& 0.245& 0.245 \\ \hline
    \end{tabular}
\end{table}

\subsection{Preliminary Results: Sequence Submodular Maximization}
\label{subsec:preliminary_results}
In this section, we introduce some additional notations and present some important preliminary results that will be used in the proofs of Theorems \ref{theorem:robustGreedy_tau_1_case_a}, \ref{theorem:robustGreedy_tau_contigous_case_a}, and \ref{theorem:robustGreedy_general_tau_case_a}. 

We begin with the additional notations. We use sequence $S^\ast(\V, k, \tau)$ to denote an optimal solution to Problem~\eqref{eq:Robust_sequence_selection}. Note that Problem~\eqref{eq:sequence_selection} is a special case of Problem~\eqref{eq:Robust_sequence_selection} with $\tau = 0$. Therefore, sequence $S^\ast(\V, k, 0)$ denotes an optimal solution to Problem~\eqref{eq:sequence_selection}. For any given sequence $S \in \HH(\V)$, we use $\Z_{\tau}(S)$ to denote an optimal solution to the following minimization problem:
\begin{equation}
\label{eq:Z_S}
 \underset{\V^\prime  \subseteq \V(S),~ |\V^\prime| \leq \tau}{\min} \quad   h(S - \V^\prime).
\end{equation}
Let $g_{\tau}(S)$ denote the value of function $h$ with input $S$ that has elements in $\Z_{\tau}(S)$ removed, i.e., $g_{\tau}(S) \triangleq h(S - \Z_{\tau}(S))$. Therefore, the optimal value of Problem~\eqref{eq:Robust_sequence_selection} is $g_{\tau}(S^\ast(\V, k, \tau)) = h(S^\ast(\V, k, \tau) - \Z_{\tau}(S^\ast(\V, k, \tau)))$. Moreover, it is easy to see that we have
\begin{equation}
\label{eq:g_vs_h}
    g_{\tau}(S^\ast(\V, k, \tau)) \leq h(S^\ast(\V, k - \tau, 0)).
\end{equation}
This is because the left-hand side is the value of a feasible solution to Problem~\eqref{eq:sequence_selection} while the right-hand side is the value of an optimal solution to Problem~\eqref{eq:sequence_selection}.

\begin{algorithm}[t]
        \caption{Sequence Submodular Greedy (SSG)}
        \label{alg:classicGreedy}
        \begin{algorithmic}[1]
            \State {Input: elements $\V$, $k$; Output: $S$}
            \State {Initialization: $S = ()$}
            \While {$|S| < k$}
            \State {$S = S \oplus \argmax_{v \in \V \setminus \V(S)} h((v)| S)$}
            \EndWhile
        \end{algorithmic}
\end{algorithm}

First, we restate the approximation performance of the SSG algorithm (presented in Algorithm~\ref{alg:classicGreedy}) in \cite{streeter2009online}. Let sequence $S$ with $|S|=k$ be the one selected by the SSG algorithm, and let $S^i$ be the sequence consisting of the first $i$ elements of sequence $S$ with $1 \leq i \leq k$. The result in \cite{streeter2009online} is stated for the case of $i = k$ only, but it can be easily generalized for any $i$ with $1 \leq i \leq k$.

 \begin{lemma}[Theorem~3 of \cite{streeter2009online}]
        \label{lemma:classic_greedy}
       Consider $1 \le i \leq k$. Under Assumption~\ref{assump:forward_backward_sequence}, we have $h(S^i) \geq (1 - 1/e^{\frac{i}{k}}) h(S^\ast(\V, k, 0))$.  
\end{lemma}

Next, we state in Lemma~\ref{lemma:value_concentration} that the approximation performance of the SSG algorithm can be better characterized if the value of the selected sequence is concentrated in its first few elements. We present the proof of Lemma~\ref{lemma:value_concentration} in Appendix~\ref{proof:value_concentration}.
\begin{lemma}
\label{lemma:value_concentration}
Consider $c \in (0, 1]$ and $1 \le k^\prime \le k$. Suppose that the sequence selected by the SSG algorithm is $S$ with $|S|=k$ and that there exists a sequence $S_1$ with $|S_1| = k - k^\prime$ such that $S_1 \preceq S$ and $h(S_1) \geq c \cdot h(S)$. Then, under Assumption~\ref{assump:forward_backward_sequence},  we have $h(S) \geq \frac{e^{ \frac{k^\prime}{k}} - 1}{e^{ \frac{k^\prime}{k}}-c}~ h(S^\ast(\V, k, 0))$.
\end{lemma}
\begin{remark}
Lemma~\ref{lemma:value_concentration} implies that additional information about the sequence selected by the SSG algorithm can be exploited to prove a better approximation ratio.
In the following, we provide a few examples to demonstrate the significance of Lemma~\ref{lemma:value_concentration}. Assume $k=20$ and $k^\prime = 19$. Then, we have $|S_1|=k-k^\prime=1$. If the value of sequence $S$ is evenly distributed over all of its twenty elements, then we have $c = 0.05$ and $h(S) \geq 0.62 h(S^\ast(\V, k, 0))$. This approximation ratio is approximately equal to that in Lemma~\ref{lemma:classic_greedy} (0.62 vs. $1-1/e \approx 0.63$). On the other hand, if the first element is worth a higher portion of the value of sequence $S$, then Lemma~\ref{lemma:value_concentration} leads to a better approximation ratio. For example, if $c = 0.5$ (resp., $c=0.8$), then we have $h(S) \geq 0.76 h(S^\ast(\V, k, 0))$ (resp., $h(S) \geq 0.88 h(S^\ast(\V, k, 0))$). In particular, if $c = 1$, then we have $h(S) \ge h(S^\ast(\V, k, 0))$, which implies that the SSG algorithm yields an optimal solution to Problem~\eqref{eq:sequence_selection}. We state this special case in Corollary~\ref{cor:c_equals_1}, which may be of independent interest. 
\end{remark}

\begin{corollary}
\label{cor:c_equals_1}
Suppose that the sequence selected by the SSG algorithm is $S$ with $|S|=k$ and that there exists a sequence $S_1$ such that $S_1 \preceq S$ and $h(S_1) = h(S)$. Then, under Assumption~\ref{assump:forward_backward_sequence},  we have $h(S) = h(S^\ast(\V, k, 0))$, and thus, sequence $S$ is an optimal solution to Problem~\eqref{eq:sequence_selection}.
\end{corollary}

\begin{remark}
Lemmas~\ref{lemma:classic_greedy} and \ref{lemma:value_concentration} and Corollary~\ref{cor:c_equals_1} can be proven in a similar manner under a weaker assumption that function $h$ is forward-monotone, backward-monotone, and element-sequence-submodular.
\end{remark}

In the following, we compare the impact of removing a certain number of elements from a selected sequence with that of removing the same number of elements from the ground set $\V$ before the selection takes place. We state the result in Lemma~\ref{lemma:Optimal_value_from_subset}.

\begin{lemma}
\label{lemma:Optimal_value_from_subset}
Consider $1 \leq \tau \leq k$. Suppose that function $h$ is forward-monotone. The following holds for any  $\V^\prime \subseteq \V$ with $|\V^\prime| \leq \tau$:
\begin{equation}
    g_{\tau}(S^\ast(\V, k, \tau)) \leq h(S^\ast(\V \setminus \V^\prime, k - \tau, 0)).
\end{equation}
\end{lemma}
\begin{proof}
Since function $h$ is forward-monotone, we can assume $|\Z_\tau(S)| = \tau$ for any sequence $S$ and $\tau \leq |S|$. The reason is the following. Suppose $|\Z_\tau(S)| < \tau$. Then, we can continue to remove more elements from the end of sequence $S$ till $|\Z_\tau(S)| = \tau$, which does not increase the value of the remaining sequence due to the forward monotonicity of function $h$.

Let $\U \triangleq \V(S^\ast(\V, k, \tau)) \cap \V^\prime$ and $\tau^\prime = |\U|$. Then, we have $\U \subseteq \V(S^\ast(\V, k, \tau))$. A little thoughts give $\U \cup \Z_{\tau - \tau^\prime}(S^\ast(\V, k, \tau) - \U)) \subseteq \V(S^\ast(\V, k, \tau))$ and $|\U \cup \Z_{\tau - \tau^\prime}(S^\ast(\V, k, \tau) - \U))| = \tau$. This implies that set $\U \cup \Z_{\tau - \tau^\prime}(S^\ast(\V, k, \tau) - \U))$ is a feasible solution to Problem~\eqref{eq:Z_S} with respect to sequence $S^\ast(\V, k, \tau)$, which further implies the following:
\begin{equation}
\label{eq:hZhU}
    h(S^\ast(\V, k, \tau) - \Z_\tau(S^\ast(\V, k, \tau)))  \leq h(S^\ast(\V, k, \tau) - \U \cup \Z_{\tau - \tau^\prime}(S^\ast(\V, k, \tau) - \U)).
\end{equation}
Also, from the definition of $\U$, we have $S^\ast(\V, k, \tau) - \U = S^\ast(\V, k, \tau) - \V^\prime$, and thus, $\Z_{\tau - \tau^\prime}(S^\ast(\V, k, \tau) - \U) = \Z_{\tau - \tau^\prime}(S^\ast(\V, k, \tau) - \V^\prime)$. Again, from the definition of $\U$, we have that the elements in $\V^\prime \setminus \U$ are not in sequence $S^\ast(\V, k, \tau)$. Then, we have the following:
\begin{equation}
\label{eq:hUhV}
    h(S^\ast(\V, k, \tau) - \U \cup \Z_{\tau - \tau^\prime}(S^\ast(\V, k, \tau) - \U))
    = h(S^\ast(\V, k, \tau) - \V^\prime \cup \Z_{\tau - \tau^\prime}(S^\ast(\V, k, \tau) - \V^\prime)).
\end{equation}
Note that sequence $S^\ast(\V, k, \tau) - \V^\prime$ does not contain any elements in $\V^\prime$ and has $k - \tau^\prime$ elements. Hence, sequence $S^\ast(\V, k, \tau) - \V^\prime$ is a feasible solution to Problem~\eqref{eq:Robust_sequence_selection} (with respect to $\V \setminus \V^\prime$, $k - \tau^\prime$, and $\tau - \tau^\prime$). This implies the following:
\begin{equation}
\label{eq:g_tau_tau_prime}
     g_{\tau - \tau^\prime}(S^\ast(\V, k, \tau) - \V^\prime) \leq g_{\tau - \tau^\prime}(S^\ast(\V \setminus \V^\prime, k - \tau^\prime, \tau - \tau^\prime)). 
\end{equation}
Also, by replacing $\V$, $k$, and $\tau$ in Eq.~\eqref{eq:g_vs_h} with $\V \setminus \V^\prime$, $k - \tau^\prime$, and $\tau - \tau^\prime$, respectively, we immediately obtain the following:
\begin{equation}
\label{eq:eq:g_0}
    g_{\tau - \tau^\prime}(S^\ast(\V \setminus \V^\prime, k - \tau^\prime, \tau - \tau^\prime)) \leq h(S^\ast(\V \setminus \V^\prime, k - \tau, 0)).
\end{equation}
By combining Eqs.~\eqref{eq:hZhU}-\eqref{eq:eq:g_0}, we have the following:
 \begin{equation}
     \begin{aligned}
        g_{\tau}(S^\ast(\V, k, \tau)) & \stackrel{\text{(a)}} = h(S^\ast(\V, k, \tau) - \Z_\tau(S^\ast(\V, k, \tau))) \\ 
        & \stackrel{\text{(b)}} \leq h(S^\ast(\V, k, \tau) - \U \cup \Z_{\tau - \tau^\prime}(S^\ast(\V, k, \tau) - \U))  \\
        & \stackrel{\text{(c)}} = h(S^\ast(\V, k, \tau) - \V^\prime \cup \Z_{\tau - \tau^\prime}(S^\ast(\V, k, \tau) - \V^\prime)) \\
        & \stackrel{\text{(d)}} = h((S^\ast(\V, k, \tau) - \V^\prime) - \Z_{\tau - \tau^\prime}(S^\ast(\V, k, \tau) - \V^\prime)) \\
        & \stackrel{\text{(e)}} = g_{\tau-\tau^\prime}((S^\ast(\V, k, \tau) - \V^\prime)) \\
        & \stackrel{\text{(f)}} \leq g_{\tau-\tau^\prime}(S^\ast(\V \setminus \V^\prime, k - \tau^\prime, \tau - \tau^\prime)) \\
        & \stackrel{\text{(g)}} \leq h(S^\ast(\V \setminus \V^\prime, k - \tau, 0)),
     \end{aligned}
    \end{equation}
where (a) is from the definition of function $g_\tau$, (b) is from Eq.~\eqref{eq:hZhU}, (c) is from Eq.~\eqref{eq:hUhV}, (d) is a rewriting, (e) is from the definition of function $g_{\tau-\tau^\prime}$, (f) is from Eq.~\eqref{eq:g_tau_tau_prime}, and (g) is from Eq.~\eqref{eq:eq:g_0}.
\end{proof}

\subsection{Proof of Lemma~\ref{lemma:value_concentration}}
\label{proof:value_concentration}

Before we prove Lemma~\ref{lemma:value_concentration}, we introduce two lemmas: Lemmas~\ref{lemma:node_marginal_value} and \ref{lemma:S1_given_S2}. Lemma~\ref{lemma:node_marginal_value} is borrowed from \cite{alaei2010maximizing} and will be used in the proof of Lemma~\ref{lemma:S1_given_S2}, which will be used in the proof of Lemma~\ref{lemma:value_concentration}.

 \begin{lemma}[Lemma~10 of \cite{alaei2010maximizing}]
    \label{lemma:node_marginal_value}
   Suppose that function $h$ is sequence-submodular. For any sequences $S_1^\prime, S_2^\prime \in \HH$, there exists an element $v \in \V$ such that $h((v)|S_1^\prime) \geq \frac{1}{|S_2^\prime|} h(S_2^\prime | S_1^\prime)$.
\end{lemma}



\begin{lemma}
\label{lemma:S1_given_S2}
Consider $1 \le k^\prime \le k$. Suppose that the sequence selected by the SSG algorithm is $S$ with $|S|=k$ and that there exist sequences $S_1$ and $S_2$ such that sequence $S$ can be written as $S = S_1 \oplus S_2$ with $|S_1| = k - k^\prime$ and $|S_2| = k^\prime$.
Then, under Assumption~\ref{assump:forward_backward_sequence}, we have $h(S_2|S_1) \geq (1 - 1/e^{\frac{k^\prime}{k}}) (h(S^\ast(\V, k, 0)) -h(S_1))$.
\end{lemma}
\begin{proof}

Let $v_2^i$ denote the $i$-th element of sequence $S_2$, and let $S_2^i \triangleq (v_2^1, \dots, v_2^i)$ denote the sequence consisting of the first $i$ elements of sequence $S_2$. Since function $h$ is forward-monotone, we can assume that $|S^\ast(\V, k, 0)| = k$ as adding more elements to the end of a sequence does not reduce its overall value. 

Due to Lemma~\ref{lemma:node_marginal_value}, there exists some element $v^\prime \in \V$ such that $h((v^\prime)|S_1 \oplus S_2^{i-1}) \geq \frac{1 }{k} h(S^\ast(\V, k, 0)|S_1 \oplus S_2^{i-1})$.
Then, we have the following:
    \begin{equation}
    \begin{aligned}
    \label{eq:element_marginal_value2}
        h(S_1 \oplus S_2^{i}) - h(S_1 \oplus S_2^{i-1})
        & = h((v_2^i)|S_1 \oplus S_2^{i-1}) \\
        & \stackrel{\text{(a)}} \ge h((v^\prime)|S_1 \oplus S_2^{i-1}) \\
        & \stackrel{\text{(b)}} \geq \frac{1 }{k} h(S^\ast(\V, k, 0)|S_1 \oplus S_2^{i-1}) \\
        & = \frac{1}{k} (h(S_1 \oplus S_2^{i-1} \oplus S^\ast(\V, k, 0)) - h(S_1 \oplus S_2^{i-1})) \\
        & \stackrel{\text{(c)}} \geq \frac{1}{k} (h(S^\ast(\V, k, 0)) - h(S_1 \oplus S_2^{i-1})),
    \end{aligned}
    \end{equation}
where (a) is due to the greedy manner of the SSG algorithm (Line~4 of Algorithm~\ref{alg:classicGreedy}), (b) is from the property of element $v^\prime$ (due to Lemma~\ref{lemma:node_marginal_value}), and (c) is due to the backward monotonicity of function $h$. Rewriting Eq.~\eqref{eq:element_marginal_value2} yields the following equivalent inequality:
     \begin{equation}
     \label{eq:recursive_i_2}
        h(S_1 \oplus S_2^{i})  \geq  \frac{1}{k} h(S^\ast(\V, k, 0)) + (1-\frac{1}{k})h(S_1 \oplus S_2^{i-1}).
    \end{equation}

By writing Eq.~\eqref{eq:recursive_i_2} for $i \in \{1, \dots, k^\prime\}$ and combining them, we obtain the following:
    \begin{equation}
     \begin{aligned}
     \label{eq:S1_plus_S_2}
        h(S_1 \oplus S_2^{k^\prime})  & \geq  \sum_{j =0}^{k^\prime - 1} \frac{1}{ k} (1-\frac{1 }{k})^j ~h(S^\ast(\V, k, 0)) + (1-\frac{1 }{k})^{k^\prime} h(S_1) \\
        & = (1 - (1 - \frac{1 }{k})^{k^\prime}) h(S^\ast(\V, k, 0)) + (1-\frac{1 }{k})^{k^\prime} h(S_1).
     \end{aligned}
    \end{equation}

Applying Eq.~\eqref{eq:S1_plus_S_2} and the fact that $S_2 = S_2^{k^\prime}$ yields the following:
 \begin{equation}
     \begin{aligned}
       h(S_2|S_1) & = h(S_1 \oplus S_2^{k^\prime}) -h(S_1) \\
        & \geq (1 - (1 - \frac{1 }{k})^{k^\prime}) h(S^\ast(\V, k, 0))  + (1-\frac{1 }{k})^{k^\prime} h(S_1) -h(S_1) \\
        & = (1 - (1 - \frac{1 }{k})^{k^\prime}) h(S^\ast(\V, k, 0)) - (1 - (1 - \frac{1 }{k})^{k^\prime})h(S_1) \\
        & = (1 - (1 - \frac{1 }{k})^{k^\prime}) (h(S^\ast(\V, k, 0)) - h(S_1))\\
        & \ge (1 - 1/e^{\frac{k^\prime}{k}}) ( h(S^\ast(\V, k, 0)) - h(S_1)),
     \end{aligned}
    \end{equation}
     where the last inequality holds because $(1 - \frac{1 }{k}) \leq e^{-\frac{1 }{k}}$ and $h(S^\ast(\V, k, 0)) - h(S_1)$ is nonnegative. This completes the proof.
\end{proof}

\begin{remark}
We note a subtle yet critical difference between the monotonicity properties of set functions and sequence functions, which complicates the proof of Lemma~\ref{lemma:S1_given_S2}. Consider a set function $r(\cdot)$. We define $r(\V_2|\V_1) \triangleq r(\V_1 \cup \V_2) - r(\V_1)$ as the marginal value of adding set $\V_2 \subseteq \V$ to another set $\V_1 \subseteq \V$.
It is easy to see that the monotonicity of function $r(\cdot)$ implies the monotonicity of function $r(\cdot|\V_1)$ for any given $\V_1 \subseteq \V$. Such an analogy also exists for forward monotonicity of sequence functions, i.e., the forward monotonicity of a sequence function $h(\cdot)$ implies the forward monotonicity of function $h(\cdot|S_1)$ for any given $S_1 \in \HH(\V)$. However, a similar property does not hold for backward monotonicity. Due to lack of such a monotonicity property, the proof of Lemma~\ref{lemma:S1_given_S2} becomes more involved and requires more careful derivations.
\end{remark}

Having introduced Lemmas~\ref{lemma:node_marginal_value} and \ref{lemma:S1_given_S2}, we are now ready to prove Lemma~\ref{lemma:value_concentration}.

\begin{proof}[Proof of Lemma~\ref{lemma:value_concentration}]
Suppose $h(S) = \delta \cdot h(S^\ast(\V, k, 0))$ for some $\delta \in (0, 1]$. 
Then, we have
 \begin{equation}
    \label{eq:h_s}
     \begin{aligned}
       \delta \cdot h(S^\ast(\V, k, 0)) & = h(S) \\
       & = h(S_1) + h(S_2|S_1) \\ 
            & \stackrel{\text{(a)}} \geq h(S_1) + (1 - 1/e^{\frac{k^\prime}{k}}) ( h(S^\ast(\V, k, 0)) - h(S_1)) \\
            & =(1/e^{\frac{k^\prime}{k}}) h(S_1) +  (1 - 1/e^{\frac{k^\prime}{k}}) h(S^\ast(\V, k, 0)) \\
             & \stackrel{\text{(b)}} \geq (c \cdot \delta /e^{\frac{k^\prime}{k}})  h(S^\ast(\V, k, 0)) +  (1 - 1/e^{\frac{k^\prime}{k}}) h(S^\ast(\V, k, 0)),
     \end{aligned}
    \end{equation}
where (a) follows from Lemma~\ref{lemma:S1_given_S2} and (b) holds because $h(S_1) \geq c \cdot h(S) = c \cdot \delta \cdot h(S^\ast(\V, k, 0))$. Dividing both sides of Eq.~\eqref{eq:h_s} by $h(S^\ast(\V, k, 0))$ yields the following:
 \begin{equation}
       \delta \geq (c \cdot \delta /e^{\frac{k^\prime}{k}})  +  (1 - 1/e^{\frac{k^\prime}{k}}),
    \end{equation}
which implies
\begin{equation}
            \delta \geq  \frac{e^{ \frac{k^\prime}{k}} - 1}{e^{ \frac{k^\prime}{k}} - c}.
\end{equation}
The above equation, along with $h(S) = \delta \cdot h(S^\ast(\V, k, 0))$, implies $h(S) \geq \frac{e^{ \frac{k^\prime}{k}} - 1}{e^{ \frac{k^\prime}{k}} - c} h(S^\ast(\V, k, 0))$. This completes the proof.
\end{proof}

\subsection{Proof of Theorem \ref{theorem:robustGreedy_tau_1_case_a}}
\label{proof:robustGreedy_tau_1_case_a}
\begin{proof}
Suppose that function $h$ is forward-monotone, backward-monotone, and sequence-submodular (Assumption~\ref{assump:forward_backward_sequence}). We use Lemmas~\ref{lemma:classic_greedy}, \ref{lemma:value_concentration}, and \ref{lemma:Optimal_value_from_subset} presented in Appendix~\ref{subsec:preliminary_results} to prove that Algorithm~\ref{alg:robustGreedy_tau_contiguous} achieves an approximation ratio of
$\max \left\{ \frac{e-1}{2e}, \frac{e^{\frac{k-2}{k-1}} - 1}{2 e^{\frac{k-2}{k-1}} - 1} \right\}$
in the case of $\tau = 1$. 

Given $\tau = 1$, in Step 1 of Algorithm~\ref{alg:robustGreedy_tau_contiguous}, the selected sequence $S_1$ consists of one element only; this element is denoted by $v_1$, i.e., $S_1 = (v_1)$. In Step 2 of Algorithm~\ref{alg:robustGreedy_tau_contiguous}, it is equivalent that sequence $S_2$ is selected by the SSG algorithm from set $\V \setminus \{v_1\}$, and we have $|S_2| = k-\tau=k-1$. Hence, the sequence selected by Algorithm~\ref{alg:robustGreedy_tau_contiguous} can be written as $S = S_1 \oplus S_2 = (v_1) \oplus S_2$. Recall that for any given sequence $S$, set $\Z_{\tau}(S)$ denotes the set of elements removed from sequence $S$ in the worst case (i.e., $\Z_{\tau}(S)$ is an optimal solution to Problem~\eqref{eq:Z_S}).
Note that only one element will be removed from $S$, i.e., $|\Z_\tau(S)|=1$. For ease of notation, we use $z$ to denote the only element in $\Z_\tau(S)$, i.e., $\Z_\tau(S) = \{z\}$.

We want to show the following two bounds, which establish the approximation ratio of Algorithm~\ref{alg:robustGreedy_tau_contiguous}:
\begin{subequations}
\begin{align}
    h(S - \{z\}) &\ge \frac{e - 1}{2 e} g_{\tau}(S^\ast(\V, k, \tau)), \label{eq:tau1bound1} \\ 
    h(S - \{z\}) &\ge \frac{e^{\frac{k-2}{k-1}} - 1}{2e^{\frac{k-2}{k-1}} - 1} g_{\tau}(S^\ast(\V, k, \tau)). \label{eq:tau1bound2}
\end{align}
\end{subequations}

To begin with, we present a lower bound on $h(S_2)$, which will be used throughout the proof:
\begin{equation}
\label{eq:h_S_2_main_tau_1}
        \begin{aligned}
         h(S_2) & \geq (1 - 1/e) h(S^\ast(\V \setminus \{v_1\}, k - \tau, 0)) \\ 
        & \geq (1 - 1/e) g_{\tau}(S^\ast(\V, k, \tau)),
    \end{aligned}
    \end{equation}
where the first inequality is from Lemma~\ref{lemma:classic_greedy} (where we replace $\V$ with $\V \setminus \{v_1\}$ and both $k$ and $i$ with $k-\tau$) and the second inequality is from Lemma~\ref{lemma:Optimal_value_from_subset} (where we replace $\V^\prime$ with $\{v_1\}$).

The proof proceeds as follows. Element $z$ is an element that will be removed, which can be either $v_1$ or an element in $S_2$. Therefore, we consider two cases: (I) $z = v_1$ and (II) $z \neq v_1$. 

In Case I, we have $z = v_1$, which implies the following: 
\begin{equation}
\label{eq:s_z_case_I_tau_1}
    \begin{aligned}
        h(S - \{z\})  & = h(S_2) \geq (1-1/e) g_{\tau}(S^\ast(\V, k, \tau)),
    \end{aligned}
\end{equation}
where the inequality follows from Eq.~\eqref{eq:h_S_2_main_tau_1}. 

In Case II, we have $z \neq v_1$ (or $z \in \V(S_2)$). Depending on the impact of removing element $z$, we consider two subcases: (II-a) $h(S_2) \leq h(S_2 - \{z\})$ and (II-b) $h(S_2) > h(S_2 - \{z\})$.
 
In Case II-a: we have $h(S_2) \leq h(S_2 - \{z\})$. In this case, the removal of element $z$ does not reduce the overall value of the remaining sequence $S_2 - \{z\}$. Then, we have
\begin{equation}
 \label{eq:s_z_case_II_a_tau_1}  
        h(S - \{z\}) = h((v_1) \oplus (S_2 - \{z\})) \stackrel{\text{(a)}} \geq h(S_2 - \{z\})
        \stackrel{\text{(b)}} \geq   h(S_2) 
        \stackrel{\text{(c)}} \geq  (1-1/e) g_{\tau}(S^\ast(\V, k, \tau)),
\end{equation}
where (a) is due to the backward monotonicity of function $h$, (b) holds from the condition of this subcase, and (c) follows from Eq.~\eqref{eq:h_S_2_main_tau_1}.

In Case II-b: we have $h(S_2) > h(S_2 - \{z\})$. Suppose $k=2$. Then, it is trivial that the sequence selected by Algorithm~\ref{alg:robustGreedy_tau_contiguous} (i.e., $S=(v_1) \oplus (z)$) yields an optimal solution. This is because removing element $z$ from $S$ gives $(v_1)$, which has the largest individual value among all the elements. Therefore, we assume $k>2$ throughout the rest of the proof. Let $\eta \triangleq \frac{h(S_2) - h(S_2 - \{z\})}{h(S_2)}$ denote the ratio of the loss of removing element $z$ from sequence $S_2$ to the value of sequence $S_2$, and we have $\eta \in (0, 1]$ due to $h(S_2) > h(S_2 - \{z\})$. We first state the following:
\begin{subequations}
\begin{align}
    &h(S - \{z\})  \ge \max \{\eta \cdot h(S_2), (1 - \eta) \cdot h(S_2)\} \label{eq:s_z_case_II_b_tau_1_assumption_a},\\
    & \max\{\eta, (1 - \eta) \} \geq \frac{1}{2} \label{eq:max_two_terms_assumption_a},\\
    & h(S_2) \geq  \frac{e^{\frac{k-2}{k-1}} - 1}{e^{\frac{k-2}{k-1}} - \eta} g_{\tau}(S^\ast(\V, k, \tau)) \label{eq:S_2_tau_1_large_k},\\
    &  \max \left\{ \eta \cdot \frac{e^{\frac{k-2}{k-1}} - 1}{e^{\frac{k-2}{k-1}} - \eta}, (1 - \eta) \cdot \frac{e^{\frac{k-2}{k-1}} - 1}{e^{\frac{k-2}{k-1}} - \eta} \right\} \geq \frac{e^{\frac{k-2}{k-1}} - 1}{2 (e^{\frac{k-2}{k-1}} -\frac{1}{2})} \label{eq:max_two_terms_assumption_a_2}.
\end{align}
\end{subequations}

We will prove Eqs.~\eqref{eq:s_z_case_II_b_tau_1_assumption_a}-\eqref{eq:max_two_terms_assumption_a_2} later; for now, we assume that they all hold. 
Then, we can obtain the following bound:
\begin{equation}
 \label{eq:s_z_case_II_b_tau_1_any_k}
    \begin{aligned}
        h(S - \{z\}) 
        & \geq \max\{\eta \cdot h(S_2),   (1 - \eta) \cdot h(S_2)\} \\
        & \geq \max\{\eta, 1 - \eta \} \cdot (1 - 1/e) g_{\tau}(S^\ast(\V, k, \tau)) \\
        & \geq \frac{e - 1}{2 e} g_{\tau}(S^\ast(\V, k, \tau)),
    \end{aligned}
\end{equation}
where the three inequalities are from Eqs.~\eqref{eq:s_z_case_II_b_tau_1_assumption_a}, \eqref{eq:h_S_2_main_tau_1}, and \eqref{eq:max_two_terms_assumption_a}, respectively. 

Similarly, we can also obtain the following bound:
\begin{equation}
   \label{eq:s_z_case_II_b_tau_1_large_k}
    \begin{aligned}
       h(S - \{z\})
       & \geq \max\{\eta \cdot h(S_2),   (1 - \eta) \cdot h(S_2)\} \\
       & \geq \max \left \{ \frac{\eta(e^{\frac{k-2}{k-1}} - 1)}{e^{\frac{k-2}{k-1}} - \eta}, \frac{(1 - \eta)(e^{\frac{k-2}{k-1}} - 1)}{e^{\frac{k-2}{k-1}} - \eta} \right\} g_{\tau}(S^\ast(\V, k, \tau)) \\
       & \geq \frac{e^{\frac{k-2}{k-1}} - 1}{2 (e^{\frac{k-2}{k-1}} -\frac{1}{2})} g_{\tau}(S^\ast(\V, k, \tau)),
       \end{aligned}
   \end{equation}
where the three inequalities are from Eqs.~\eqref{eq:s_z_case_II_b_tau_1_assumption_a}, \eqref{eq:S_2_tau_1_large_k}, and \eqref{eq:max_two_terms_assumption_a_2}, respectively.

Combining all the cases establishes an approximation ratio of Algorithm~\ref{alg:robustGreedy_tau_contiguous} and completes the proof. Specifically, combining the bounds in Eqs.~\eqref{eq:s_z_case_I_tau_1}, \eqref{eq:s_z_case_II_a_tau_1}, and \eqref{eq:s_z_case_II_b_tau_1_any_k} (resp., \eqref{eq:s_z_case_II_b_tau_1_large_k}) yields the bound in Eq.~\eqref{eq:tau1bound1} (resp., \eqref{eq:tau1bound2}).


Now, it remains to show that Eqs.~\eqref{eq:s_z_case_II_b_tau_1_assumption_a}-\eqref{eq:max_two_terms_assumption_a_2} hold in Case II-b, where we have $z \in \V(S_2)$ and $h(S_2) > h(S_2 - \{z\})$.
We first rewrite $S_2$ as $S_2=S_2^1 \oplus (z) \oplus S_2^2$, where $S_2^1$ and $S_2^2$ denote the subsequences of $S_2$ before and after element $z$, respectively. Note that $S_2^1$ or $S_2^2$ could be an empty sequence, depending on the position of $z$ in $S_2$. Then, we characterize $h((z))$ in terms of $h(S_2)$:
 \begin{equation}
 \label{eq:z2_tau_1_assumption_a}
    \begin{aligned}
        \eta \cdot h(S_2) & \stackrel{\text{(a)}} = h(S_2) - h(S_2 - \{z\}) \\
        & = h(S_2^1 \oplus (z) \oplus S_2^2) - h(S_2^1 \oplus S_2^2) \\
        & = h(S_2^1) + h((z)|S_2^1) +  h(S_2^2|S_2^1 \oplus (z))  - h(S_2^1) -  h(S_2^2|S_2^1) \\
        & = h((z)|S_2^1) +  h(S_2^2|S_2^1 \oplus (z)) -  h(S_2^2|S_2^1) \\
        & \leq h((z)|S_2^1) \\
        & \leq h((z)),
    \end{aligned}
 \end{equation}
  where (a) is from the definition of $\eta$ and the two inequalities are due to the sequence submodularity of function $h$. We are now ready to prove Eqs.~\eqref{eq:s_z_case_II_b_tau_1_assumption_a}-\eqref{eq:max_two_terms_assumption_a_2}.

To prove Eq.~\eqref{eq:s_z_case_II_b_tau_1_assumption_a}, we decompose it into two parts: (i) $h(S - \{z\})  \ge \eta \cdot h(S_2)$ and (ii) $h(S - \{z\}) \ge (1 - \eta) \cdot h(S_2)$.

Part (i) can be shown through the following:
\begin{equation}
        h(S - \{z\}) \stackrel{\text{(a)}} \geq h((v_1))  \stackrel{\text{(b)}} \geq h((z))  \stackrel{\text{(c)}} \geq \eta \cdot h(S_2),
\end{equation}
where (a) is form the forward monotonicity of function $h$, (b) is due to the greedy manner of Algorithm~\ref{alg:robustGreedy_tau_contiguous} (Lines~3-5), and (c) is from Eq.~\eqref{eq:z2_tau_1_assumption_a}.

Part (ii) can be shown through the following:
\begin{equation*}
   h(S - \{z\}) 
   = h((v_1) \oplus (S_2 - \{z\})) 
   \stackrel{\text{(a)}} \geq  h(S_2 - \{z\})  \stackrel{\text{(b)}} =  (1-\eta) \cdot h(S_2), 
 \end{equation*}
 where (a) is from the backward monotonicity of function $h$ and (b) is from the definition of $\eta$.
 
Eq.~\eqref{eq:max_two_terms_assumption_a} holds trivially for any $\eta \in (0, 1]$ by setting $\eta$ and $1-\eta$ to be equal and solving for $\eta$.

Next, we show that Eq.~\eqref{eq:S_2_tau_1_large_k} holds. Let $v^1_2$ denote the first element of sequence $S_2$. Then, we have the following: 
\begin{equation*}
    \begin{aligned}
        h((v^1_2)) & \stackrel{\text{(a)}} \geq h((z)) \stackrel{\text{(b)}}  \geq \eta \cdot h(S_2),
    \end{aligned}
\end{equation*}
where (a) holds because element $v^1_2$ has the largest individual value among all elements in $S_2$ and (b) follows from Eq.~\eqref{eq:z2_tau_1_assumption_a}. Then, we can characterize the value of $h(S_2)$ as follows:
 \begin{equation*}
    \begin{aligned}
        h(S_2) & \geq \frac{e^{\frac{k-2}{k-1}} - 1}{e^{\frac{k-2}{k-1}} -\eta} h(S^\ast(\V \setminus \{v_1\}, k - \tau, 0)) \\
        & \geq  \frac{e^{\frac{k-2}{k-1}} - 1}{e^{\frac{k-2}{k-1}} -\eta} g_{\tau}(S^\ast(\V, k, \tau)),
    \end{aligned}
 \end{equation*}
\noindent where the first inequality is from Lemma~\ref{lemma:value_concentration} (where we replace $\V$, $S$, $S_1$, $k$, $k^\prime$, and $c$ with $\V \setminus \{v_1\}$, $S_2$, $(v^1_2)$, $k-1$, $k-2$, and $\eta$, respectively) and the second inequality is from Lemma~\ref{lemma:Optimal_value_from_subset}.

Finally, we show that Eq.~\eqref{eq:max_two_terms_assumption_a_2} holds. We define two functions of $\eta$:
    $l_1(\eta) \triangleq \eta \cdot \frac{e^{\frac{k-2}{k-1}} - 1}{e^{\frac{k-2}{k-1}} - \eta}$ and $l_2(\eta) \triangleq (1 - \eta) \cdot \frac{e^{\frac{k-2}{k-1}} - 1}{e^{\frac{k-2}{k-1}} - \eta}$. 
It is easy to verify that for $k > 2$ and $\eta \in (0, 1]$, function $l_1(\eta)$ is monotonically increasing and function $l_2(\eta)$ is monotonically decreasing. Also, we have $l_1(\frac{1}{2}) = l_2(\frac{1}{2}) = \frac{e^{\frac{k-2}{k-1}} - 1}{2 (e^{\frac{k-2}{k-1}} -\frac{1}{2})}$. We consider two cases for $\eta$: $\eta \in [\frac{1}{2}, 1]$ and $\eta \in (0, \frac{1}{2}]$. For $\eta \in [\frac{1}{2}, 1]$, we have $\max\{l_1(\eta), l_2(\eta)\} \geq l_1(\eta) \geq l_1(\frac{1}{2})$ as $l_1(\eta)$ is monotonically increasing; for $\eta \in (0, \frac{1}{2}]$, we have $\max\{l_1(\eta), l_2(\eta)\} \geq l_2(\eta) \geq l_2(\frac{1}{2}) = l_1(\frac{1}{2})$ as $l_2(\eta)$ is monotonically decreasing. Therefore, for $\eta \in (0, 1]$, we have $\max\{l_1(\eta), l_2(\eta)\} \geq l_1(\frac{1}{2}) = \frac{e^{\frac{k-2}{k-1}} - 1}{2 (e^{\frac{k-2}{k-1}} -\frac{1}{2})}$. This gives Eq.~\eqref{eq:max_two_terms_assumption_a_2} and completes the proof.
\end{proof}

\subsection{Proof of Theorem \ref{theorem:robustGreedy_tau_contigous_case_a}}
\label{proof:robustGreedy_tau_contigous_case_a}
\begin{proof}
Suppose that function $h$ is forward-monotone, backward-monotone, and sequence-submodular (Assumption~\ref{assump:forward_backward_sequence}). We use Lemmas~\ref{lemma:classic_greedy}, \ref{lemma:value_concentration}, and \ref{lemma:Optimal_value_from_subset} presented in Appendix~\ref{subsec:preliminary_results} to prove that Algorithm~\ref{alg:robustGreedy_tau_contiguous} achieves an approximation ratio of $\max \left \{ \frac{(e - 1)^2}{e(2e-1)}, \frac{(e - 1)(e^{\frac{k-2\tau}{k-\tau}} -1)}{(2e - 1)e^{\frac{k-2\tau}{k-\tau}}- (e - 1)} \right \}$ in the case of $1 \leq \tau \leq k$, assuming the removal of $\tau$ contiguous elements. 

In Step 1 of Algorithm~\ref{alg:robustGreedy_tau_contiguous}, it is equivalent that sequence $S_1$ is selected by the SSG algorithm from set $\V$, and we have $|S_1| = \tau$. Similarly, in Step 2 of Algorithm~\ref{alg:robustGreedy_tau_contiguous}, it is also equivalent that sequence $S_2$ is selected by the SSG algorithm from set $\V \setminus \V(S_1)$, and we have $|S_2| = k-\tau$. Hence, the sequence selected by Algorithm~\ref{alg:robustGreedy_tau_contiguous} can be written as $S = S_1 \oplus S_2$. Recall that for any given sequence $S$, set $\Z_{\tau}(S)$ denotes the set of elements removed from sequence $S$ in the worst case (i.e., $\Z_{\tau}(S)$ is an optimal solution to Problem~\eqref{eq:Z_S}). We define $\Z^1_\tau(S) \triangleq \Z_\tau(S) \cap \V(S_1)$ and $\Z^2_\tau(S) \triangleq \Z_\tau(S) \cap \V(S_2)$ as the set of elements removed from subsequences $S_1$ and $S_2$, respectively.

The proof of Theorem~\ref{theorem:robustGreedy_tau_contigous_case_a} follows a similar line of analysis as in the proof of Theorem~\ref{theorem:robustGreedy_tau_1_case_a} for the case of $\tau=1$. Specifically, we will also consider three cases: (I) $\Z^2_\tau(S) = \emptyset$, (II-a) $\Z^2_\tau(S) \neq \emptyset$ and $h(S_2) \leq h(S_2 - \Z^2_\tau(S))$, and (II-b) $\Z^2_\tau(S) \neq \emptyset$ and $h(S_2) > h(S_2 - \Z^2_\tau(S))$. The proofs of Case I and Case II-a are almost the same as those in Theorem~\ref{theorem:robustGreedy_tau_1_case_a}, except for some minor technical differences. However, the proof of Case II-b is substantially different. The reason is the following. In the proof of Theorem~\ref{theorem:robustGreedy_tau_1_case_a}, only one element can be removed from $S$, and in Case II-b, this element has to be in $\V(S_2)$, which makes it easier to characterize the impact of removing such an element. In the case of $1 \leq \tau \leq k$, however, more than one element may be removed, which could be in either $\V(S_1)$ or $\V(S_2)$ or both. Therefore, we present a different approach to address this new technical challenge.

We want to show the following two bounds that establish the approximation ratio of Algorithm~\ref{alg:robustGreedy_tau_contiguous}:

\begin{subequations}
\begin{align}
    h(S - \Z_\tau(S)) &\ge \frac{(e - 1)^2}{e(2e-1)} g_{\tau}(S^\ast(\V, k, \tau)), \label{eq:tau_contigous_bound1} \\
    h(S - \Z_\tau(S)) &\ge \frac{(e - 1)(e^{\frac{k-2\tau}{k-\tau}} -1)}{(2e - 1)e^{\frac{k-2\tau}{k-\tau}}- (e - 1)}g_{\tau}(S^\ast(\V, k, \tau)). \label{eq:tau_contigous_bound2}
\end{align}
\end{subequations}

To begin with, we present a lower bound on $h(S_2)$, which will be used throughout the proof:
\begin{equation}
\label{eq:h_S_2_main_tau_contigous}
        \begin{aligned}
         h(S_2) & \geq (1 - 1/e) h(S^\ast(\V \setminus \V(S_1), k - \tau, 0)) \\ 
        & \geq (1 - 1/e) g_{\tau}(S^\ast(\V, k, \tau)),
    \end{aligned}
    \end{equation}
where the first inequality is from Lemma~\ref{lemma:classic_greedy} (where we replace $\V$, $k$, and $i$ with $\V \setminus \V(S_1)$, $k-\tau$, and $k-\tau$, respectively) and the second inequality is from Lemma~\ref{lemma:Optimal_value_from_subset} (where we replace $\V^\prime$ with $\V(S_1)$).

The proof proceeds as follows. Elements in $\Z_\tau(S)$ will be removed from sequence $S$. These elements can be either fully or partially in $\V(S_1)$ (i.e., $\Z^2_\tau(S) = \emptyset$ or $\Z^2_\tau(S) \neq \emptyset$). Therefore, we consider two cases: (I) $\Z^2_\tau(S) = \emptyset$ and (II) $\Z^2_\tau(S) \neq \emptyset$. 

In Case I, we have $\Z^2_\tau(S) = \emptyset$, i.e., $\Z_\tau(S) = \V(S_1)$. Then, we have the following: 
\begin{equation}
\label{eq:s_z_case_I_tau_congtigous}
    \begin{aligned}
        h(S - \Z_\tau(S)) & = h(S_2)  \geq (1 - 1/e) g_\tau(S^\ast(\V, k, \tau)),
    \end{aligned}
\end{equation}
where the inequality follows from Eq.~\eqref{eq:h_S_2_main_tau_contigous}. 

In Case II, we have $\Z^2_\tau(S) \neq \emptyset$.
Depending on the impact of removing elements in $\Z^2_\tau(S)$, we consider two subcases: (II-a) $h(S_2) \leq h(S_2 - \Z^2_\tau(S))$ and (II-b) $h(S_2) > h(S_2 - \Z^2_\tau(S))$.

In Case II-a: we have $h(S_2) \leq h(S_2 - \Z^2_\tau(S))$. In this case, the removal of elements in $\Z^2_\tau(S)$ does not reduce the overall value of the remaining sequence $S_2 - \Z^2_\tau(S)$. Then, we have
\begin{equation}
 \label{eq:s_z_case_II_a_tau_congtigous}  
    \begin{aligned}
        h(S - \Z_\tau(S)) & = h((S_1 - \Z^1_\tau(S)) \oplus (S_2 - \Z^2_\tau(S))) \\
        & \stackrel{\text{(a)}} \geq h(S_2 - \Z^2_\tau(S)) \\
        & \stackrel{\text{(b)}} \geq   h(S_2) \\
        & \stackrel{\text{(c)}} \geq  (1 - 1/e) g_\tau(S^\ast(\V, k, \tau)),
    \end{aligned}
\end{equation}
where (a) is due to the backward monotonicity of function $h$, (b) holds from the condition of this subcase, and (c) follows from Eq.~\eqref{eq:h_S_2_main_tau_contigous}.

In Case II-b: we have $h(S_2) > h(S_2 - \Z^2_\tau(S))$. Let $\tau_1 \triangleq |\Z^1_\tau(S)|$ and $\tau_2 \triangleq |\Z^2_\tau(S)|$. Then, we have $\tau=\tau_1+\tau_2$ and $k = |S| = |S_1| + |S_2| \ge \tau + \tau_2$. We consider two cases: $k = \tau + \tau_2$ and $k > \tau + \tau_2$. 

Suppose $k = \tau + \tau_2$. Then, it implies $\Z^2_\tau(S) = \V(S_2)$, i.e., all the elements in $S_2$ are removed. This further implies that the elements in $\Z^1_\tau(S)$ are at the end of sequence $S_1$ (due to the contiguous assumption of elements in $\Z_\tau(S)$). Let $S_1^{\tau_2} \preceq S_1$ denote the subsequence consisting of the first $\tau_2$ elements in $S_1$. It is easy to see $S_1 - \Z^1_\tau(S) = S_1^{\tau_2}$. Then, we have the following:
\begin{equation}
\label{eq:h_S_1}
\begin{aligned}
         h(S - \Z_\tau(S)) & = h((S_1 - \Z^1_\tau(S)) \oplus (S_2 - \Z^2_\tau(S))) \\
         & = h(S_1 - \Z^1_\tau(S)) \\
         & = h(S_1^{\tau_2}) \\
         & \geq (1 - 1/e) h(S^\ast(\V, \tau_2, 0)) \\ 
        & \geq (1 - 1/e) g_{\tau}(S^\ast(\V, k, \tau)),
    \end{aligned}
\end{equation}
where the first inequality is from Lemma~\ref{lemma:classic_greedy} (where we replace both $k$ and $i$ with $\tau_2$) and the second inequality is due to $\tau_2=k-\tau$ and Lemma~\ref{lemma:Optimal_value_from_subset} (where $\V^\prime$ is an empty set).

Now, suppose $k>\tau + \tau_2$. Let $\eta \triangleq \frac{h(S_2) - h(S_2 - \Z^2_\tau(S))}{h(S_2)}$ denote the ratio of the loss of removing elements in $\Z^2_\tau(S)$ from sequence $S_2$ to the value of sequence $S_2$, and we have $\eta \in (0, 1]$ due to $h(S_2) > h(S_2 - \Z^2_\tau(S))$. We first state the following:
\begin{subequations}
\begin{align}
    & h(S - \Z_\tau(S)) \geq \max\{\eta (1 - 1/e) h(S_2),  (1 - \eta) h(S_2)\} \label{eq:s_z_case_II_b_tau_contigous_assumption_a},\\
    & \max\{\eta (1 - 1/e),  (1 - \eta)\} \geq \frac{ e  - 1}{2e - 1}     \label{eq:max_two_terms_tau_contigous_assumption_a},\\
    & h(S_2) \geq \frac{e^{\frac{k-\tau - \tau_2}{k-\tau_2}} - 1}{e^{\frac{k-\tau - \tau_2}{k-\tau_2}} - \eta (1 - 1/e)} g_\tau(S^\ast(\V, k, \tau))  \label{eq:S_2_large_k_tau_contigous_assumption_a},\\
    & \max \left\{\frac{\eta(1 - 1/e)(e^{\frac{k-\tau - \tau_2}{k-\tau_2}} - 1)}{e^{\frac{k-\tau - \tau_2}{k-\tau_2}} - \eta (1 - 1/e)},  \frac{(1 - \eta)(e^{\frac{k-\tau - \tau_2}{k-\tau_2}} - 1)}{e^{\frac{k-\tau - \tau_2}{k-\tau_2}} - \eta (1 - 1/e)} \right\} \geq \frac{(e - 1)(e^{\frac{k-\tau - \tau_2}{k-\tau_2}} -1)}{(2e - 1)e^{\frac{k-\tau - \tau_2}{k-\tau_2}}- (e - 1)} \label{eq:max_two_terms_tau_contigous_2_assumption_a}.
\end{align}
\end{subequations}

We will prove Eqs.~\eqref{eq:s_z_case_II_b_tau_contigous_assumption_a}-\eqref{eq:max_two_terms_tau_contigous_2_assumption_a} later; for now, we assume that they all hold. Then, we can obtain the following bound:
\begin{equation}
 \label{eq:s_z_case_II_b_tau_congtigous_any_k}
    \begin{aligned}
        h(S - \Z_\tau(S))&\geq \max\{\eta (1 - 1/e) h(S_2),   (1 - \eta) h(S_2)\} \\
        & \geq \max\{\eta (1 - 1/e), (1 - \eta) \} \cdot (1 - 1/e) g_\tau(S^\ast(\V, k, \tau)) \\
        & \geq \frac{ (e - 1)^2}{e(2e-1)} g_\tau(S^\ast(\V, k, \tau)),
    \end{aligned}
\end{equation}
where the three inequalities are from Eqs.~\eqref{eq:s_z_case_II_b_tau_contigous_assumption_a}, \eqref{eq:h_S_2_main_tau_contigous}, and \eqref{eq:max_two_terms_tau_contigous_assumption_a}, respectively.  

Similarly, we can also obtain the following bound:
\begin{equation}
   \label{eq:s_z_case_II_b_tau_congtigous_large_k}
    \begin{aligned}
       h(S - \Z_\tau(S)) &\geq \max\{\eta(1 - 1/e) h(S_2),   (1 - \eta) h(S_2)\} \\
       & \geq \max \left\{\frac{\eta(e - 1)(e^{\frac{k-\tau - \tau_2}{k-\tau_2}} - 1)}{e(e^{\frac{k-\tau - \tau_2}{k-\tau_2}} - \eta (1 - 1/e))},  \frac{(1 - \eta)(e^{\frac{k-\tau - \tau_2}{k-\tau_2}} - 1)}{e^{\frac{k-\tau - \tau_2}{k-\tau_2}} - \eta (1 - 1/e)} \right\} g_\tau(S^\ast(\V, k, \tau)) \\
       & \geq \frac{(e - 1)(e^{\frac{k-\tau - \tau_2}{k-\tau_2}} -1)}{(2e - 1)e^{\frac{k-\tau - \tau_2}{k-\tau_2}}- (e - 1)} g_\tau(S^\ast(\V, k, \tau)) \\
       & \geq \frac{(e - 1)(e^{\frac{k-2\tau}{k-\tau}} -1)}{(2e - 1)e^{\frac{k-2\tau}{k-\tau}}- (e - 1)} g_\tau(S^\ast(\V, k, \tau)),
       \end{aligned}
   \end{equation}
where the first three inequalities are from Eqs.~\eqref{eq:s_z_case_II_b_tau_contigous_assumption_a}, \eqref{eq:S_2_large_k_tau_contigous_assumption_a}, and \eqref{eq:max_two_terms_tau_contigous_2_assumption_a}, respectively, and the last inequality is due to $\tau_2 \leq \tau$ and that $\frac{(e - 1)(e^{\frac{k-\tau - \tau_2}{k-\tau_2}} -1)}{(2e - 1)e^{\frac{k-\tau - \tau_2}{k-\tau_2}}- (e - 1)}$ is a decreasing function of $\tau_2$. 

Combining all the cases establishes the approximation ratios of Algorithm~\ref{alg:robustGreedy_tau_contiguous} and completes the proof. Specifically, combining the bounds in Eqs.~\eqref{eq:s_z_case_I_tau_congtigous}, \eqref{eq:s_z_case_II_a_tau_congtigous}, \eqref{eq:h_S_1}, and \eqref{eq:s_z_case_II_b_tau_congtigous_any_k} (resp., \eqref{eq:s_z_case_II_b_tau_congtigous_large_k}) yields the bound in Eq.~\eqref{eq:tau_contigous_bound1} (resp., \eqref{eq:tau_contigous_bound2}).

Now, it remains to show that Eqs.~\eqref{eq:s_z_case_II_b_tau_contigous_assumption_a}-\eqref{eq:max_two_terms_tau_contigous_2_assumption_a} hold in Case II-b, where we have $\Z^2_\tau(S) \neq \emptyset$ and $h(S_2) > h(S_2 - \Z^2_\tau(S))$. Recall that the $\tau$ elements in $\Z_\tau(S)$ form a contiguous subsequence of $S$. Then, the elements in $\Z^1_\tau(S)$ and $\Z^2_\tau(S)$ also form a contiguous subsequence of $S_1$ and $S_2$, respectively. We use $Z_1$ and $Z_2$ to denote the contiguous subsequence of elements in $\Z^1_\tau(S)$ and $\Z^2_\tau(S)$, respectively. We first rewrite $S_2$ as $S_2=S_2^1 \oplus Z_2 \oplus S_2^2$, where $S_2^1$ and $S_2^1$ denote the subsequences of $S_2$ before and after subsequence $Z_2$, respectively. Note that $S_2^1$ or $S_2^2$ could be an empty sequence, depending on the position of subsequence $Z_2$ in $S_2$. Then, we characterize $h(Z_2)$ in terms of $h(S_2)$:
 \begin{equation}
 \label{eq:z2_tau_contigous_assumption_a}
    \begin{aligned}
        \eta \cdot h(S_2) & \stackrel{\text{(a)}} = h(S_2) - h(S_2 - \Z^2_\tau(S)) \\
        & = h(S_2^1 \oplus Z_2 \oplus S_2^2) - h(S_2^1 \oplus S_2^2) \\
        & = h(S_2^1) + h(Z_2|S_2^1) +  h(S_2^2|S_2^1 \oplus Z_2)  - h(S_2^1) -  h(S_2^2|S_2^1) \\
        & = h(Z_2|S_2^1) +  h(S_2^2|S_2^1 \oplus Z_2) -  h(S_2^2|S_2^1) \\
        & \leq h(Z_2|S_2^1) \\
        & \leq h(Z_2),
    \end{aligned}
 \end{equation}
 where (a) is from the definition of $\eta$ and the two inequalities are due the sequence submodularity of function $h$. We are now ready to prove Eqs.~\eqref{eq:s_z_case_II_b_tau_contigous_assumption_a}-\eqref{eq:max_two_terms_tau_contigous_2_assumption_a}.

To prove Eq.~\eqref{eq:s_z_case_II_b_tau_contigous_assumption_a}, we decompose it into two parts: (i) $h(S - \Z_\tau(S))  \ge \eta (1 - 1/e) h(S_2)$ and (ii) $h(S - \Z_\tau(S)) \ge (1 - \eta) h(S_2)$.

Recall that $S_1^{\tau_2}$ denote the subsequence consisting of the first $\tau_2$ elements in $S_1$. Then, Part (i) can be shown through the following: 
\begin{equation}
\label{eq:s_z_case_II_b_tau_contigous_assumption_a_1}
    \begin{aligned}
        h(S - \Z_\tau(S)) & = h((S_1 - \Z^1_\tau(S)) \oplus (S_2 - \Z^2_\tau(S))) \\
        & \stackrel{\text{(a)}} \geq h(S_1 - \Z^1_\tau(S)) \\
        & = h(S_1^{\tau_2}) \\
        & \stackrel{\text{(b)}} \geq (1 - 1/e) h(S^\ast(\V, \tau_2, 0)) \\
        & \stackrel{\text{(c)}} \geq (1 - 1/e) h(Z_2)\\
        & \stackrel{\text{(d)}} \geq \eta (1 - 1/e) h(S_2),
    \end{aligned}
 \end{equation}
\noindent where (a) is due to the forward monotonicity of function $h$, (b) is due to the greedy manner of selecting subsequence $S_1^{\tau_2}$ from set $\V$ and Lemma~\ref{lemma:classic_greedy} (where we replace both $k$ and $i$ with $\tau_2$), (c) holds because sequence $Z_2$ is a feasible solution to Problem~\eqref{eq:sequence_selection} for selecting a sequence of $\tau_2$ elements from $\V$, and (d) is from Eq.~\eqref{eq:z2_tau_contigous_assumption_a}.

Part (ii) can be shown through the following:
\begin{equation}
\label{eq:s_z_case_II_b_tau_contigous_assumption_a_2}
    \begin{aligned}
   h(S - \Z_\tau(S)) &= h((S_1 - \Z^1_\tau(S)) \oplus (S_2 - \Z^2_\tau(S))) \\ &\stackrel{\text{(a)}} \geq  h(S_2 - \Z^2_\tau(S)) \\ &\stackrel{\text{(b)}} =  (1-\eta) h(S_2), 
    \end{aligned}
 \end{equation}
 where (a) is from the backward monotonicity of function $h$ and (b) is from the definition of $\eta$.
 
Eq.~\eqref{eq:max_two_terms_tau_contigous_assumption_a} holds trivially for any $\eta \in (0, 1]$ by setting $\eta (1-1/e)$ and $1-\eta$ to be equal, solving for $\eta$, and plugging it back.

Next, we show that Eq.~\eqref{eq:S_2_large_k_tau_contigous_assumption_a} holds. Let $S_{2}^{\tau_2} \preceq S_2$ denote the subsequence consisting of the first $\tau_2$ elements of sequence $S_2$. Then, we have the following:
  \begin{equation}
    \begin{aligned}
    h(S_{2}^{\tau_2})& \stackrel{\text{(a)}} \geq (1 - 1/e) h(S^\ast(\V \setminus \V(S_1), \tau_2, 0)) \\
        & \stackrel{\text{(b)}} \geq (1 - 1/e) h(Z_2)\\
        & \stackrel{\text{(c)}} \geq \eta (1 - 1/e) h(S_2),
    \end{aligned}
 \end{equation}
where (a) is due to the greedy manner of selecting subsequence $S_{2}^{\tau_2}$ from set $\V \setminus \V(S_1)$ and Lemma~\ref{lemma:classic_greedy} (where we replace $\V$, $k$, and $i$ with $\V \setminus \V(S_1)$, $\tau_2$, and $\tau_2$, respectively), (b) holds because sequence $Z_2$ is a feasible solution to Problem \eqref{eq:sequence_selection} for selecting a sequence of $\tau_2$ elements from $\V = \V \setminus \V(S_1)$, and (c) is from Eq.~\eqref{eq:z2_tau_contigous_assumption_a}. Therefore, we can characterize the value of $h(S_2)$ as follows:
 \begin{equation}
    \begin{aligned}
        h(S_2) & \geq  \frac{e^{\frac{k-\tau - \tau_2}{k-\tau_2}} - 1}{e^{\frac{k-\tau - \tau_2}{k-\tau_2}} - \eta (1 - 1/e)} h(S^\ast(\V \setminus \V(S_1), k - \tau, 0)) \\
        & \geq \frac{e^{\frac{k-\tau - \tau_2}{k-\tau_2}} - 1}{e^{\frac{k-\tau - \tau_2}{k-\tau_2}} - \eta (1 - 1/e)} g_\tau(S^\ast(\V, k, \tau)),
    \end{aligned}
 \end{equation}
\noindent  where the first inequality is from Lemma~\ref{lemma:value_concentration} (where we replace $\V$, $S$, $S_1$, $k$, $k^\prime$, and $c$ with $\V \setminus \V(S_1)$, $S_2$, $S_{2}^{\tau_2}$, $k-\tau$, $k-\tau-\tau_2$, and $\eta (1 - 1/e)$, respectively), and the second inequality is from Lemma~\ref{lemma:Optimal_value_from_subset}.

Finally, we show that Eq.~\eqref{eq:max_two_terms_tau_contigous_2_assumption_a} holds. We define two functions of $\eta$:
\begin{equation*}
l_1(\eta) = \eta (1 - 1/e) \cdot \frac{e^{\frac{k-\tau - \tau_2}{k-\tau_2}} - 1}{e^{\frac{k-\tau - \tau_2}{k-\tau_2}} - \eta (1 - 1/e)} ~\text{and}~ l_2(\eta) = (1 - \eta) \cdot \frac{(e^{\frac{k-\tau - \tau_2}{k-\tau_2}} - 1)}{e^{\frac{k-\tau - \tau_2}{k-\tau_2}} - \eta (1 - 1/e)}.
\end{equation*}
It is easy to verify that for $k > \tau + \tau_2$ and $\eta \in (0, 1]$, function $l_1(\eta)$ is monotonically increasing and function $l_2(\eta)$ is monotonically decreasing. 
Also, we have $l_1(\frac{e}{2e - 1}) = l_2(\frac{e}{2e - 1}) = \frac{(e - 1)(e^{\frac{k-\tau - \tau_2}{k-\tau_2}} -1)}{(2e - 1)e^{\frac{k-\tau - \tau_2}{k-\tau_2}}- (e - 1)}$.
We consider two cases for $\eta$: $\eta \in [\frac{e}{2e - 1}, 1]$ and $\eta \in (0, \frac{e}{2e - 1}]$. For $\eta \in [\frac{e}{2e - 1}, 1]$, we have $\max\{l_1(\eta), l_2(\eta)\} \geq l_1(\eta) \geq l_1(\frac{e}{2e - 1})$ as $l_1(\eta)$ is monotonically increasing; for $\eta \in (0, \frac{e}{2e - 1}]$, we have $\max\{l_1(\eta), l_2(\eta)\} \geq l_2(\eta) \geq l_2(\frac{e}{2e - 1}) = l_1(\frac{e}{2e - 1})$ as $l_2(\eta)$ is monotonically decreasing. Therefore, for $\eta \in (0, 1]$, we have $\max\{l_1(\eta), l_2(\eta)\} \geq l_1(\frac{e}{2e - 1}) = \frac{(e - 1)(e^{\frac{k-\tau - \tau_2}{k-\tau_2}} -1)}{(2e - 1)e^{\frac{k-\tau - \tau_2}{k-\tau_2}}- (e - 1)}$. This gives Eq.~\eqref{eq:max_two_terms_tau_contigous_2_assumption_a} and completes the proof.
\end{proof}

\subsection{Proof of Theorem \ref{theorem:robustGreedy_general_tau_case_a}}
\label{proof:robustGreedy_general_tau_case_a}
\begin{proof}
Suppose that function $h$ is forward-monotone, backward-monotone, and general-sequence-submodular (Assumption~\ref{assump:forward_backward_general}). We use Lemmas~\ref{lemma:classic_greedy} and \ref{lemma:Optimal_value_from_subset} presented in Appendix~\ref{subsec:preliminary_results} to prove that Algorithm \ref{alg:robustGreedy_tau_general} achieves an approximation ratio of $\frac{1- 1/e}{1 + \tau}$ in the case of $1 \leq \tau \leq k$, assuming the removal of an arbitrary subset of $\tau$ selected elements, which are not necessarily contiguous.

In Step 1 of Algorithm~\ref{alg:robustGreedy_tau_general}, sequence $S_1$ is selected by choosing $\tau$ elements with the highest individual values in a greedy manner, and we have $|S_1| = \tau$. In Step 2 of Algorithm~\ref{alg:robustGreedy_tau_contiguous}, it is equivalent that sequence $S_2$ is selected by the SSG algorithm from set $\V \setminus \V(S_1)$, and we have $|S_2| = k-\tau$. Hence, the sequence selected by Algorithm~\ref{alg:robustGreedy_tau_contiguous} can be written as $S = S_1 \oplus S_2$. Recall that for any given sequence $S$, set $\Z_{\tau}(S)$ denotes the set of elements removed from sequence $S$ in the worst case (i.e., $\Z_{\tau}(S)$ is an optimal solution to Problem~\eqref{eq:Z_S}). We define $\Z^1_\tau(S) \triangleq \Z_\tau(S) \cap \V(S_1)$ and $\Z^2_\tau(S) \triangleq \Z_\tau(S) \cap \V(S_2)$ as the set of elements removed from subsequences $S_1$ and $S_2$, respectively.

The proof of Theorem~\ref{theorem:robustGreedy_general_tau_case_a} follows a similar line of analysis as in the proof of Theorem~\ref{theorem:robustGreedy_tau_contigous_case_a}. Specifically, we will also consider three cases: (I) $\Z^2_\tau(S) = \emptyset$, (II-a) $\Z^2_\tau(S) \neq \emptyset$ and $h(S_2) \leq h(S_2 - \Z^2_\tau(S))$, and (II-b) $\Z^2_\tau(S) \neq \emptyset$ and $h(S_2) > h(S_2 - \Z^2_\tau(S))$. The proofs of Case I and Case II-a are identical to those in Theorem~\ref{theorem:robustGreedy_tau_contigous_case_a}. Therefore, we focus on Case~II-b, which requires a different proof strategy.

We want to show the following bound that establishes the approximation ratio of Algorithm~\ref{alg:robustGreedy_tau_general}:
\begin{equation}
\label{eq:tau_general_bound1}
    h(S - \Z_\tau(S)) \ge \frac{1- 1/e}{1 + \tau} g_{\tau}(S^\ast(\V, k, \tau)).
\end{equation}

To begin with, we present a lower bound on $h(S_2)$, which will be used throughout the proof:
\begin{equation}
\label{eq:h_S_2_main_tau_general}
        \begin{aligned}
         h(S_2) & \geq (1 - 1/e) h(S^\ast(\V \setminus \V(S_1), k - \tau, 0)) \\ 
        & \geq (1 - 1/e) g_{\tau}(S^\ast(\V, k, \tau)),
    \end{aligned}
    \end{equation}
where the first inequality is from Lemma~\ref{lemma:classic_greedy} (where we replace $\V$, $k$, and $i$ with $\V \setminus \V(S_1)$, $k-\tau$, and $k-\tau$, respectively) and the second inequality is from Lemma~\ref{lemma:Optimal_value_from_subset} (where we replace $\V^\prime$ with $\V(S_1)$).

We now focus on Case II-b, where we have $\Z^2_\tau(S) \neq \emptyset$ and $h(S_2) > h(S_2 - \Z^2_\tau(S))$. Let $\eta \triangleq \frac{h(S_2) - h(S_2 - \Z^2_\tau(S))}{h(S_2)}$ denote the ratio of the loss of removing elements in $\Z^2_\tau(S)$ from sequence $S_2$ to the value of sequence $S_2$, and we have $\eta \in (0, 1]$ due to $h(S_2) > h(S_2 - \Z^2_\tau(S))$. We first state the following:
\begin{subequations}
\begin{align}
    & h(S - \Z_\tau(S)) \geq \max\{\frac{ \eta}{\tau} \cdot h(S_2),  (1 - \eta) \cdot h(S_2)\}, \label{eq:s_z_case_II_b_tau_general_assumption_a}\\
    & \max\{\frac{ \eta}{\tau}, (1 - \eta)\} \geq  \frac{1}{\tau + 1}.     \label{eq:max_two_terms_tau_general_assumption_a}
\end{align}
\end{subequations}

We will prove Eqs.~\eqref{eq:s_z_case_II_b_tau_general_assumption_a} and \eqref{eq:max_two_terms_tau_general_assumption_a} later; for now, we assume that they both hold. Then, we can obtain the following:
\begin{equation}
\label{eq:s_z_any_k_assumption_a_general_tau}
    \begin{aligned}
        h(S - \Z_\tau(S)) & \geq \max\{\frac{ \eta}{\tau} \cdot h(S_2),  (1 - \eta) \cdot h(S_2)\} \\
        & \geq \max\{\frac{ \eta}{\tau},  (1 - \eta)\} \cdot (1 - 1/e)g_\tau(S^\ast(\V, k, \tau))  \\
        & \geq  \frac{1 - 1/e}{\tau+1} g_\tau(S^\ast(\V, k, \tau)),
    \end{aligned}
\end{equation}
where the three inequalities are from Eqs.~\eqref{eq:s_z_case_II_b_tau_general_assumption_a}, \eqref{eq:h_S_2_main_tau_general}, and \eqref{eq:max_two_terms_tau_general_assumption_a}, respectively.  

Now, we show that Eqs.~\eqref{eq:s_z_case_II_b_tau_general_assumption_a} and \eqref{eq:max_two_terms_tau_general_assumption_a} hold. We start by characterizing the value of elements in $\Z^2_\tau(S)$. Let $\tau_2 \triangleq |\Z^2_\tau(S)|$, and let the elements in $\Z^2_\tau(S)$ be denoted by $z_1, z_2, \dots, z_{\tau_2}$ according to their order in sequence $S_2$. Then, we can rewrite $S_2$ as  $S_2 = S_2^1 \oplus (z_1) \oplus S_2^2 \oplus (z_2) \oplus \dots \oplus S_2^{\tau_2 + 1}$, where $S_2^i$ is the subsequence between elements $z_{i-1}$ and $z_i$, for $i = 1,2, \dots, \tau_2 + 1$, and both $z_0$ and $z_{\tau+1}$ are an empty sequence. Note that subsequence $S_2^i$ could be an empty sequence, for $i = 1, 2, \dots, \tau_2+1$. We characterize the value of elements in $\Z^2_\tau(S)$ in the following:
 \begin{equation}
 \label{eq:z2_general_tau}
    \begin{aligned}
        \eta \cdot h(S_2)  \stackrel{\text{(a)}}=& h(S_2) - h(S_2 - \Z^2_\tau(S)) \\
        =& h(S_2^1) + h((z_1)|S_2^1) +  h(S_2^2|S_2^1 \oplus (z_1)) + \dots   \\
        & + h(S_2^{\tau_2+1}|S_2^1 \oplus (z_1) \oplus \dots \oplus S_2^{\tau_2} \oplus (z_{\tau_2}))  \\
        & -  h(S_2^1) - h(S_2^2|S_2^1) - \dots   - h(S_2^{\tau_2+1}|S_2^1 \oplus \dots \oplus S_2^{\tau_2}) \\
        =& \sum_{i =1}^{\tau_2} h((z_i)|S_2^1 \oplus (z_1) \oplus \dots \oplus (z_{i-1}) \oplus S_2^{i})   + h(S_2^2|S_2^1 \oplus (z_1)) - h(S_2^2|S_2^1) \\
        & + \dots + h(S_2^{\tau_2+1}|S_2^1 \oplus (z_1) \oplus \dots \oplus S_2^{\tau_2} \oplus (z_{\tau_2}))  - 
        h(S_2^{\tau_2+1}|S_2^1 \oplus \dots \oplus S_2^{\tau_2}) \\
        \leq& \sum_{i =1}^{\tau_2} h((z_i)|S_2^1 \oplus (z_1) \oplus \dots \oplus (z_{i-1}) \oplus S_2^{i}) \\
        \leq& \sum_{i =1}^{\tau_2} h((z_i)),
    \end{aligned}
 \end{equation}
 where (a) is from the definition of $\eta$ and the two inequalities are due to the general sequence submodularity of function $h$.

To prove Eq.~\eqref{eq:s_z_case_II_b_tau_general_assumption_a}, we decompose it into two parts: (i) $h(S - \Z_\tau(S))  \ge \frac{\eta}{\tau} \cdot h(S_2)$ and (ii) $h(S - \Z_\tau(S)) \ge (1 - \eta) \cdot h(S_2)$.

Let $v^\prime$ denote the first element in $S_1 - \Z^1_\tau(S)$. Then, Part (i) can be shown through the following:
 \begin{equation*}
        h(S - \Z_\tau(S)) \stackrel{\text{(a)}} \geq h((v^\prime)) 
        \stackrel{\text{(b)}} \geq \frac{1}{\tau_2}  \sum_{i =1}^{\tau_2} h((z_i))
        \stackrel{\text{(c)}} \geq \frac{\eta}{\tau_2} \cdot  h(S_2)
        \stackrel{\text{(d)}} \geq \frac{\eta}{\tau} \cdot h(S_2),
 \end{equation*}  
where (a) is due to the forward monotonicity of function $h$, (b) is due to $h(v^\prime) \geq h(z_i)$ for any $i = 1,2, \dots, \tau_2$, (c) is from Eq.~\eqref{eq:z2_general_tau}, and (d) is due to $\tau_2 \leq \tau$.

Part (ii) can be shown through the following:
\begin{equation*}
    \begin{aligned}
   h(S - \Z_\tau(S)) = h((S_1 - \Z^1_\tau(S)) \oplus (S_2 - \Z^2_\tau(S))) \stackrel{\text{(a)}} \geq h(S_2 - \Z^2_\tau(S)) \stackrel{\text{(b)}} = (1-\eta) \cdot h(S_2),
    \end{aligned}
 \end{equation*}
 where (a) is from the backward monotonicity of function $h$ and (b) is from the definition of $\eta$.
 
 Eq.~\eqref{eq:max_two_terms_tau_general_assumption_a} holds trivially for any $\eta \in (0, 1]$ by setting $\frac{\eta}{\tau}$ and $1-\eta$ to be equal, solving for $\eta$, and plugging it back. This completes the proof.
\end{proof}

\subsection{Preliminary Results: Approximate Sequence Submodular Maximization}
\label{subsec:approximate_preliminary_results}
In this section, we generalize the results of Lemmas~\ref{lemma:classic_greedy} and \ref{lemma:value_concentration} under weaker assumptions based on approximate versions of sequence submodularity and backward monotonicity. Specifically, we introduce two weaker assumptions that will be used in this section. 



\begin{assumption}
\label{assump:forward_alpha_mu1}
Function $h$ is forward-monotone, $\alpha$-backward-monotone, and $\mu_1$-element-sequence-submodular.
\end{assumption}

\begin{assumption}
\label{assump:forward_backward_mu1}
Function $h$ is forward-monotone, backward-monotone, and $\mu_1$-element-sequence-submodular.
\end{assumption}

We first generalize the approximation result of the SSG algorithm (Algorithm~\ref{alg:classicGreedy}) under Assumption~\ref{assump:forward_alpha_mu1}. Let sequence $S$ with $|S|=k$ be the one selected by the SSG algorithm, let $S^i$ be the sequence consisting of the first $i$ elements of sequence $S$ with $1 \leq i \leq k$. We borrow the result of \cite{sallam2019placement} and restate a generalized version in Lemma~\ref{lemma:approximate_classic_greedy}. This is also a generalization of Lemma~\ref{lemma:classic_greedy}. While \cite{sallam2019placement} assumes that function $h$ is forward-monotone, $\alpha$-backward-monotone, and element-sequence-submodular and considers the case of $i = k$ only, the generalization to the result in Lemma~\ref{lemma:approximate_classic_greedy} is fairly straightforward. Therefore, we refer the interested reader to \cite{sallam2019placement} for the proof.

 \begin{lemma}[Theorem~1 of \cite{sallam2019placement}]
        \label{lemma:approximate_classic_greedy}
       Consider $1 \le i \leq k$. Under Assumption~\ref{assump:forward_alpha_mu1}, we have $h(S^i) \geq \alpha (1 - \frac{1}{e^{ \mu_1 \cdot \frac{i}{k}}}) h(S^\ast(\V, k, 0))$.  
\end{lemma}

We also generalize the result of Lemma~\ref{lemma:value_concentration} and present the generalized result in Lemma~\ref{lemma:approximate_value_concentration}. Note that the current proof techniques we use could only lead to the generalization under Assumption~\ref{assump:forward_backward_mu1}, which is slightly stronger than Assumption~\ref{assump:forward_alpha_mu1} (backward-monotone vs. $\alpha$-backward-monotone). 

\begin{lemma}
\label{lemma:approximate_value_concentration}
Consider $c \in (0, 1]$ and $1 \le k^\prime \le k$. Suppose that the sequence selected by the SSG algorithm is $S$ with $|S|=k$ and that there exists a sequence $S_1$ with $|S_1| = k - k^\prime$ such that $S_1 \preceq S$ and $h(S_1) \geq c \cdot h(S)$. Then, under Assumption~\ref{assump:forward_backward_mu1},  we have $h(S) \geq \frac{e^{ \mu_1 \cdot \frac{ k^\prime}{k}} - 1}{e^{ \mu_1 \cdot \frac{k^\prime}{k}}-c}~ h(S^\ast(\V, k, 0))$.
\end{lemma}

Before we prove Lemma~\ref{lemma:approximate_value_concentration}, we introduce two lemmas: Lemmas~\ref{lemma:approximate_node_marginal_value} and \ref{lemma:approximate_S1_given_S2}. Lemma~\ref{lemma:approximate_node_marginal_value} is a generalization of Lemma~\ref{lemma:node_marginal_value} and will be used in the proof of Lemma~\ref{lemma:approximate_S1_given_S2}, which will be used in the proof of Lemma~\ref{lemma:approximate_value_concentration}.

 \begin{lemma}
    \label{lemma:approximate_node_marginal_value}
   Suppose that function $h$ is $\mu_1$-element-sequence-submodular. For any sequences $S_1^\prime, S_2^\prime \in \HH$, there exists an element $v \in \V$ such that $h((v)|S_1^\prime) \geq \frac{\mu_1}{|S_2^\prime|} h(S_2^\prime | S_1^\prime)$.
\end{lemma}
\begin{proof}
Let $S_2^\prime = (u_1, \dots, u_{|S_2^\prime|})$. We can rewrite $h(S_2^\prime|S_1^\prime)$ as
         \begin{equation}
         \label{eq:sum_S1_terms}
            h(S_2^\prime|S_1^\prime) =  \sum_{j=1}^{{|S_2^\prime|}} h((u_j)|S_1^\prime \oplus (u_1, \dots, u_{j-1})),
    \end{equation} 
where sequence $(u_1, \dots, u_{j-1})$ is an empty sequence when $j = 1$. 
Due to the Pigeonhole principle, there exists some $j^\prime \in \{1, 2, \dots |S^\prime|\}$ such that the following is satisfied:
    \begin{equation}
         \begin{aligned}
         \label{eq:max_term}
            h((u_{j^\prime})|S_1^\prime \oplus (u_1, \dots, u_{j^\prime -1 })) 
            &\geq \frac{1}{|S_2^\prime|} \sum_{j=1}^{{|S_2^\prime|}} h((u_j)|S_1^\prime \oplus (u_1, \dots, u_{j-1}))\\
            &= \frac{1}{|S_2^\prime|}  h(S_2^\prime|S_1^\prime),
            \end{aligned}
    \end{equation} 
    where the equality is from Eq.~\eqref{eq:sum_S1_terms}.
    Furthermore, element $u_{j^\prime}$ satisfies the following:
    \begin{equation}
    \begin{aligned}
             h((u_{j^\prime})|S_1^\prime) & \stackrel{\text{(a)}} {\geq} \mu_1 \cdot h((u_{j^\prime})|S_1^\prime \oplus (u_1, \dots, u_{j^\prime -1 })) \\
             & \stackrel{\text{(b)}} \geq \frac{\mu_1}{|S_2^\prime|}  h(S_2^\prime|S_1^\prime),
             \end{aligned}
    \end{equation} 
    where (a) is due to function $h$ being $\mu_1$-element-sequence-submodular and (b) is from Eq.~\eqref{eq:max_term}. This completes the proof.
\end{proof}

\begin{lemma}
\label{lemma:approximate_S1_given_S2}
Consider $1 \le k^\prime \le k$. Suppose that the sequence selected by the SSG algorithm is $S$ with $|S|=k$ and that there exist sequences $S_1$ and $S_2$ such that sequence $S$ can be written as $S = S_1 \oplus S_2$ with $|S_1| = k - k^\prime$ and $|S_2| = k^\prime$. Then, under Assumption~\ref{assump:forward_backward_mu1}, we have $h(S_2|S_1) \geq (1 - \frac{1}{e^{\mu_1 \cdot \frac{k^\prime}{k}}}) (h(S^\ast(\V, k, 0)) -h(S_1))$.
\end{lemma}
\begin{proof}
The proof follows a similar line of analysis as in Lemma~\ref{lemma:S1_given_S2}.
Let $v_i$ denote the $i$-th element of sequence $S_2$, and let $S_2^i \triangleq (v_1, \dots, v_i)$ denote the sequence consisting of the first $i$ elements of sequence $S_2$. Since function $h$ is forward-monotone, we can assume that $|S^\ast(\V, k, 0)| = k$ as adding more elements to the end of a sequence does not reduce its overall value. 

Due to Lemma~\ref{lemma:approximate_node_marginal_value}, there exists some element $v^\prime \in \V$ such that $h((v^\prime)|S_1 \oplus S_2^{i-1}) \geq \frac{\mu_1}{k} h(S^\ast(\V, k, 0)|S_1 \oplus S_2^{i-1})$.
Then, we have the following:
    \begin{equation}
    \begin{aligned}
    \label{eq:approximate_element_marginal_value2}
        h(S_1 \oplus S_2^{i}) - h(S_1 \oplus S_2^{i-1})
        & = h((v_i)|S_1 \oplus S_2^{i-1}) \\
        & \stackrel{\text{(a)}} \ge h((v^\prime)|S_1 \oplus S_2^{i-1}) \\
        & \stackrel{\text{(b)}} \geq \frac{\mu_1}{k} h(S^\ast(\V, k, 0)|S_1 \oplus S_2^{i-1}) \\
        & = \frac{\mu_1}{k} (h(S_1 \oplus S_2^{i-1} \oplus S^\ast(\V, k, 0)) - h(S_1 \oplus S_2^{i-1})) \\
        & \stackrel{\text{(c)}} \geq \frac{\mu_1}{k} (h(S^\ast(\V, k, 0)) - h(S_1 \oplus S_2^{i-1})),
    \end{aligned}
    \end{equation}
where (a) is due to the greedy manner of the SSG algorithm (Line~4 of Algorithm~\ref{alg:classicGreedy}), (b) is from the property of element $v^\prime$ (due to Lemma~\ref{lemma:approximate_node_marginal_value}), and (c) is due to the backward monotonicity of function $h$. Rewriting Eq.~\eqref{eq:approximate_element_marginal_value2} yields the following equivalent inequality:
 
     \begin{equation}
     \label{eq:approximate_recursive_i_2}
        h(S_1 \oplus S_2^{i})  \geq  \frac{\mu_1}{k} h(S^\ast(\V, k, 0)) + (1-\frac{\mu_1}{k})h(S_1 \oplus S_2^{i-1}).
    \end{equation}

By writing Eq.~\eqref{eq:approximate_recursive_i_2} for $i \in \{1, \dots, k^\prime\}$ and combining them, we obtain the following:
    \begin{equation}
     \begin{aligned}
     \label{eq:approximate_S1_plus_S_2}
        h(S_1 \oplus S_2^{k^\prime})  & \geq  \sum_{j =0}^{k^\prime - 1} \frac{\mu_1}{ k} (1-\frac{\mu_1}{k})^j ~h(S^\ast(\V, k, 0)) + (1-\frac{\mu_1}{k})^{k^\prime} h(S_1) \\
        & = (1 - (1 - \frac{\mu_1}{k})^{k^\prime}) h(S^\ast(\V, k, 0)) + (1-\frac{\mu_1}{k})^{k^\prime} h(S_1).
     \end{aligned}
    \end{equation}

Applying Eq.~\eqref{eq:approximate_S1_plus_S_2} and the fact that $S_2 = S_2^{k^\prime}$ yields the following:
 \begin{equation}
     \begin{aligned}
       h(S_2|S_1) & = h(S_1 \oplus S_2^{k^\prime}) -h(S_1) \\
        & \geq (1 - (1 - \frac{\mu_1}{k})^{k^\prime}) h(S^\ast(\V, k, 0))  + (1-\frac{\mu_1}{k})^{k^\prime} h(S_1) -h(S_1) \\
        & = (1 - (1 - \frac{\mu_1}{k})^{k^\prime}) h(S^\ast(\V, k, 0)) - (1 - (1 - \frac{\mu_1}{k})^{k^\prime})h(S_1) \\
        & = (1 - (1 - \frac{\mu_1}{k})^{k^\prime}) (h(S^\ast(\V, k, 0)) - h(S_1))\\
        & \ge (1 - \frac{1}{e^{\mu_1 \cdot \frac{k^\prime}{k}}}) (h(S^\ast(\V, k, 0)) - h(S_1)),
     \end{aligned}
    \end{equation}
     where the last inequality holds because $(1 - \frac{\mu_1}{k}) \leq e^{-\frac{\mu_1}{k}}$ and $h(S^\ast(\V, k, 0)) - h(S_1)$ is nonnegative. This completes the proof.
\end{proof}
Having introduced Lemmas~\ref{lemma:approximate_node_marginal_value} and \ref{lemma:approximate_S1_given_S2}, we are now ready to prove Lemma~\ref{lemma:approximate_value_concentration}.

\begin{proof}[Proof of Lemma~\ref{lemma:approximate_value_concentration}]

Suppose $h(S) = \delta \cdot h(S^\ast(\V, k, 0))$ for some $\delta \in (0, 1]$. Then, we have
 \begin{equation}
    \label{eq:approximate_h_s}
     \begin{aligned}
       \delta \cdot h(S^\ast(\V, k, 0)) & = h(S) \\
       & = h(S_1) + h(S_2|S_1) \\ 
            & \stackrel{\text{(a)}} \geq h(S_1) + (1 - \frac{1}{e^{\mu_1 \cdot \frac{k^\prime}{k}}}) ( h(S^\ast(\V, k, 0)) - h(S_1)) \\
            & =\frac{1}{e^{\mu_1 \cdot \frac{ k^\prime}{k}}} h(S_1) +  (1 - \frac{1}{e^{\mu_1 \cdot \frac{k^\prime}{k}}}) h(S^\ast(\V, k, 0)) \\
             & \stackrel{\text{(b)}} \geq \frac{c \delta}{e^{\mu_1 \cdot \frac{ k^\prime}{k}}}  h(S^\ast(\V, k, 0)) +  (1 - \frac{1}{e^{\mu_1 \cdot \frac{k^\prime}{k}}}) h(S^\ast(\V, k, 0)),
     \end{aligned}
    \end{equation}
where (a) is from Lemma~\ref{lemma:approximate_S1_given_S2} and (b) is due to $h(S_1) \geq c \cdot h(S) = c \cdot \delta \cdot h(S^\ast(\V, k, 0))$. Dividing both sides of Eq.~\eqref{eq:approximate_h_s} by $h(S^\ast(\V, k, 0))$ yields the following:
 \begin{equation}
       \delta \geq \frac{c \delta}{e^{\mu_1 \cdot \frac{ k^\prime}{k}}}  +  1 - \frac{1}{e^{\mu_1 \cdot \frac{k^\prime}{k}}},
    \end{equation}
which implies
\begin{equation}
            \delta \geq  \frac{e^{\mu_1 \cdot  \frac{k^\prime}{k}} - 1}{e^{ \mu_1 \cdot \frac{k^\prime}{k}} - c}.
\end{equation}
Since $h(S) = \delta \cdot h(S^\ast(\V, k, 0))$, we have $h(S) \geq \frac{e^{\mu_1 \cdot  \frac{k^\prime}{k}} - 1}{e^{ \mu_1 \cdot \frac{k^\prime}{k}} - c} h(S^\ast(\V, k, 0))$. This completes the proof.
\end{proof}

\subsection{Proof of Theorem \ref{theorem:robustGreedy_tau_1_case_b}}
\label{proof:robustGreedy_tau_1_case_b}
\begin{proof}
We use Lemma~\ref{lemma:Optimal_value_from_subset} presented in Appendix~\ref{subsec:preliminary_results} and Lemmas~\ref{lemma:approximate_classic_greedy} and \ref{lemma:approximate_value_concentration} presented in Appendix~\ref{subsec:approximate_preliminary_results} to prove the approximation ratios of Algorithm~\ref{alg:robustGreedy_tau_contiguous} in the case of $\tau = 1$. The proof follows a similar line of analysis as in the proof of Theorem~\ref{theorem:robustGreedy_tau_1_case_a}.

Given $\tau = 1$, in Step 1 of Algorithm~\ref{alg:robustGreedy_tau_contiguous}, the selected sequence $S_1$ consists of one element only; this element is denoted by $v_1$, i.e., $S_1 = (v_1)$. In Step 2 of Algorithm~\ref{alg:robustGreedy_tau_contiguous}, it is equivalent that sequence $S_2$ is selected by the SSG algorithm from set $\V \setminus \{v_1\}$, and we have $|S_2| = k-\tau=k-1$. Hence, the sequence selected by Algorithm~\ref{alg:robustGreedy_tau_contiguous} can be written as $S = S_1 \oplus S_2 = (v_1) \oplus S_2$. Recall that for any given sequence $S$, set $\Z_{\tau}(S)$ denotes the set of elements removed from sequence $S$ in the worst case (i.e., $\Z_{\tau}(S)$ is an optimal solution to Problem~\eqref{eq:Z_S}).
Note that only one element will be removed from $S$, i.e., $|\Z_\tau(S)|=1$. For ease of notation, we use $z$ to denote the only element in $\Z_\tau(S)$, i.e., $\Z_\tau(S) = \{z\}$.

We want to show the following two bounds, which establish the approximation ratios of Algorithm~\ref{alg:robustGreedy_tau_contiguous}: 
1) Suppose that function $h$ is forward-monotone, backward-monotone,  $\mu_1$-element-sequence-submodular, and $\mu_2$-sequence-submodular (Assumption~\ref{assump:forward_backward_mu1_mu2}), we have
\begin{equation}
    h(S - \{z\}) \ge \frac{a(e^{b} - 1)}{e^{b} - a} g_{\tau}(S^\ast(\V, k, \tau)), ~\text{where}~ a = \frac{\mu_1 \mu_2}{\mu_1 + 1} ~\text{and}~ b = \frac{\mu_1 (k -2)}{k-1}; \label{eq:approximate_tau1bound1} \\ 
\end{equation}
2) Suppose that function $h$ is forward-monotone, $\alpha$-backward-monotone,  $\mu_1$-element-sequence-submodular, and $\mu_2$-sequence-submodular (Assumption~\ref{assump:forward_alpha_mu1_mu2}), we have
\begin{equation}
\label{eq:approximate_tau1bound2}
    h(S - \{z\}) \ge \frac{\alpha^2 \mu_1 \mu_2 (e^{\mu_1} - 1)}{(\mu_1  + \alpha  )e^{\mu_1}} g_{\tau}(S^\ast(\V, k, \tau)). 
\end{equation}

We begin with the proof of Eq.~\eqref{eq:approximate_tau1bound2} as some of the intermediate results will be used in the proof of Eq.~\eqref{eq:approximate_tau1bound1}.
We first present a lower bound on $h(S_2)$, which will be used throughout the proof:
\begin{equation}
\label{eq:approximate_h_S_2_main_tau_1}
        \begin{aligned}
         h(S_2) & \geq \alpha (1 - 1/e^{ \mu_1}) h(S^\ast(\V \setminus \{v_1\}, k - \tau, 0)) \\ 
        & \geq \alpha (1 - 1/e^{ \mu_1}) g_{\tau}(S^\ast(\V, k, \tau)),
    \end{aligned}
    \end{equation}
where the first inequality is from Lemma~\ref{lemma:approximate_classic_greedy} (where we replace $\V$, $k$, and $i$ with $\V \setminus \{v_1\}$, $k-\tau$, and $k-\tau$, respectively) and the second inequality is from Lemma~\ref{lemma:Optimal_value_from_subset} (where we replace $\V^\prime$ with $\{v_1\}$).

The proof proceeds as follows. Element $z$ is an element that will be removed, which can be either $v_1$ or an element in $S_2$. Therefore, we consider two cases: (I) $z = v_1$ and (II) $z \neq v_1$. 

In Case I, we have $z = v_1$, which implies the following: 
\begin{equation}
\label{eq:approximate_s_z_case_I_tau_1}
    \begin{aligned}
        h(S - \{z\})  & = h(S_2) \geq \alpha (1-1/e^{ \mu_1}) g_{\tau}(S^\ast(\V, k, \tau)),
    \end{aligned}
\end{equation}
where the inequality follows from Eq.~\eqref{eq:approximate_h_S_2_main_tau_1}. 

In Case II, we have $z \neq v_1$ (or $z \in \V(S_2)$). Depending on the impact of removing element $z$, we consider two subcases: (II-a) $h(S_2) \leq h(S_2 - \{z\})$ and (II-b) $h(S_2) > h(S_2 - \{z\})$.
 
In Case II-a: we have $h(S_2) \leq h(S_2 - \{z\})$. In this case, the removal of element $z$ does not reduce the overall value of the remaining sequence $S_2 - \{z\}$. Then, we have
\begin{equation}
 \label{eq:approximate_s_z_case_II_a_tau_1}  
    \begin{aligned}
        h(S - \{z\})  & = h((v_1) \oplus (S_2 - \{z\})) \\
        & \stackrel{\text{(a)}} \geq \alpha \cdot h(S_2 - \{z\}) \\
        & \stackrel{\text{(b)}} \geq   \alpha \cdot h(S_2) \\
        & \stackrel{\text{(c)}} \geq  \alpha^2 \cdot (1-1/e^{ \mu_1}) g_{\tau}(S^\ast(\V, k, \tau)),
    \end{aligned}
\end{equation}
where (a) is due to the $\alpha$-backward monotonicity of function $h$, (b) holds from the condition of this subcase, and (c) follows from Eq.~\eqref{eq:approximate_h_S_2_main_tau_1}.

In Case II-b: we have $h(S_2) > h(S_2 - \{z\})$. Let $\eta \triangleq \frac{h(S_2) - h(S_2 - \{z\})}{h(S_2)}$ denote the ratio of the loss of removing element $z$ from sequence $S_2$ to the value of sequence $S_2$, and we have $\eta \in (0, 1]$ due to $h(S_2) > h(S_2 - \{z\})$. We first state the following:
\begin{subequations}
\begin{align}
    &h(S - \{z\})  \ge \max \{ \mu_1(\eta + \mu_2 -1) \cdot h(S_2), \alpha (1 - \eta) \cdot h(S_2)\} \label{eq:approximate_s_z_case_II_b_tau_1_assumption_a},\\
    & \max\{\mu_1(\eta + \mu_2 -1), \alpha (1 - \eta) \} \geq \frac{\alpha \mu_1 \mu_2}{\mu_1  + \alpha  } \label{eq:approximate_max_two_terms_assumption_a}.
\end{align}
\end{subequations}

We will prove Eqs.~\eqref{eq:approximate_s_z_case_II_b_tau_1_assumption_a} and \eqref{eq:approximate_max_two_terms_assumption_a} later; for now, we assume that they both hold. 
Then, we can obtain the following bound:
\begin{equation}
 \label{eq:approximate_s_z_case_II_b_tau_1_any_k}
    \begin{aligned}
        h(S - \{z\}) 
        & \geq \max\{\mu_1(\eta + \mu_2 -1) \cdot h(S_2),   \alpha (1 - \eta) \cdot h(S_2)\} \\
        & \geq \max\{\mu_1(\eta + \mu_2 -1), \alpha (1 - \eta) \} \cdot \alpha  (1 - 1/e^{ \mu_1}) g_{\tau}(S^\ast(\V, k, \tau)) \\
        & \geq \frac{\alpha^2 \mu_1  \mu_2 (e^{ \mu_1} - 1)}{(\mu_1  + \alpha )e^{ \mu_1}} g_{\tau}(S^\ast(\V, k, \tau)),
    \end{aligned}
\end{equation}
where the three inequalities are from Eqs.~\eqref{eq:approximate_s_z_case_II_b_tau_1_assumption_a}, \eqref{eq:approximate_h_S_2_main_tau_1}, and \eqref{eq:approximate_max_two_terms_assumption_a}, respectively.

By combining the bounds in Eqs.~\eqref{eq:approximate_s_z_case_I_tau_1}, \eqref{eq:approximate_s_z_case_II_a_tau_1}, and \eqref{eq:approximate_s_z_case_II_b_tau_1_any_k}, we obtain the bound in Eq.~\eqref{eq:approximate_tau1bound2}.  

Next, we prove Eq.~\eqref{eq:approximate_s_z_case_II_b_tau_1_assumption_a}. We first rewrite $S_2$ as $S_2=S_2^1 \oplus (z) \oplus S_2^2$, where $S_2^1$ and $S_2^2$ denote the subsequences of $S_2$ before and after element $z$, respectively. Note that $S_2^1$ or $S_2^2$ could be an empty sequence, depending on the position of $z$ in $S_2$. Then, we characterize $h((z))$ in terms of $h(S_2)$:
 \begin{equation}
 \label{eq:intermediate_approximate_z2_tau_1_assumption_a}
    \begin{aligned}
        \eta \cdot h(S_2) & \stackrel{\text{(a)}} = h(S_2) - h(S_2 - \{z\}) \\
        & = h(S_2^1 \oplus (z) \oplus S_2^2) - h(S_2^1 \oplus S_2^2) \\
        & = h(S_2^1) + h((z)|S_2^1) +  h(S_2^2|S_2^1 \oplus (z))  - h(S_2^1) -  h(S_2^2|S_2^1) \\
        & = h((z)|S_2^1) +  h(S_2^2|S_2^1 \oplus (z)) -  h(S_2^2|S_2^1) \\
        & \stackrel{\text{(b)}} \leq h((z)|S_2^1) +  h(S_2^2|S_2^1 \oplus (z)) -  \mu_2 h(S_2^2|S_2^1 \oplus (z)) \\
        & = h((z)|S_2^1) +   ( 1 -  \mu_2) h(S_2^2|S_2^1 \oplus (z)) \\
        & \stackrel{\text{(c)}} \leq h((z)|S_2^1) +   ( 1 -  \mu_2) h(S_2) \\
        & \stackrel{\text{(d)}} \leq h((z))/\mu_1 + ( 1 -  \mu_2) h(S_2),
    \end{aligned}
 \end{equation}
  where (a) is from the definition of $\eta$, (b) is  due to function $h$ being $\mu_2$-sequence-submodular, (c) holds because $\mu_2 \leq 1$ and $h(S_2^2|S_2^1 \oplus (z)) \leq h(S_2^2|S_2^1 \oplus (z)) + h(S_2^1 \oplus (z)) = h(S_2)$ (recall that $S_2 = S_2^1\oplus (z) \oplus S_2^2$), and (d) is due to function $h$ being $\mu_1$-element-sequence-submodular.
From Eq.~\eqref{eq:intermediate_approximate_z2_tau_1_assumption_a}, we have the following: 
\begin{equation}
    \label{eq:approximate_z2_tau_1_assumption_a}
    h((z)) \geq \mu_1(\eta + \mu_2 -1) h(S_2).
\end{equation}

To prove Eq.~\eqref{eq:approximate_s_z_case_II_b_tau_1_assumption_a}, we decompose it into two parts: (i) $h(S - \{z\})  \ge \mu_1(\eta + \mu_2 -1) \cdot h(S_2)$ and (ii) $h(S - \{z\}) \ge \alpha (1-\eta) \cdot  h(S_2)$.

Part (i) can be shown through the following:
\begin{equation}
        h(S - \{z\}) \stackrel{\text{(a)}} \geq h((v_1))  \stackrel{\text{(b)}} \geq h((z))  \stackrel{\text{(c)}} \geq \mu_1(\eta + \mu_2 -1) \cdot h(S_2),
\end{equation}
where (a) is from the forward monotonicity of function $h$, (b) is due to the greedy manner of Algorithm~\ref{alg:robustGreedy_tau_contiguous} (Lines~3-5), and (c) is from Eq.~\eqref{eq:approximate_z2_tau_1_assumption_a}.

Part (ii) can be shown through the following:
\begin{equation*}
   h(S - \{z\}) 
   = h((v_1) \oplus (S_2 - \{z\})) 
   \stackrel{\text{(a)}} \geq  \alpha \cdot h(S_2 - \{z\})  \stackrel{\text{(b)}} =  \alpha (1-\eta) \cdot h(S_2), 
 \end{equation*}
 where (a) is from the $\alpha$-backward monotonicity of function $h$ and (b) is from the definition of $\eta$. 

Eq.~\eqref{eq:approximate_max_two_terms_assumption_a} holds trivially for any $\eta \in (0, 1]$ by setting $\mu_1(\eta + \mu_2 -1)$ and $\alpha (1 - \eta)$ to be equal, solving for $\eta$, and plugging it back. This completes the proof of the bound in Eq.~\eqref{eq:approximate_tau1bound2}.

Next, we prove the bound in Eq.~\eqref{eq:approximate_tau1bound1}. Note that the analysis so far applies to function $h$ that is forward-monotone, $\alpha$-backward-monotone,  $\mu_1$-element-sequence-submodular, and $\mu_2$-sequence-submodular (Assumption~\ref{assump:forward_alpha_mu1_mu2}), and hence, it also applies to  Assumption~\ref{assump:forward_backward_mu1_mu2}, which is a special case of Assumption~\ref{assump:forward_alpha_mu1_mu2} with $\alpha = 1$. The bound in Eq.~\eqref{eq:approximate_tau1bound1} requires backward monotonicity (i.e., $\alpha = 1$), but it becomes better than the bound in Eq.~\eqref{eq:approximate_tau1bound2} when $k$ is large. In the following analysis, we assume that function $h$ is forward-monotone, backward-monotone,  $\mu_1$-element-sequence-submodular, and $\mu_2$-sequence-submodular (Assumption~\ref{assump:forward_backward_mu1_mu2}).

The proof proceeds as follows. We borrow the analysis of Case-I, Case~II-a, and Eq.~\eqref{eq:approximate_s_z_case_II_b_tau_1_assumption_a} from the previous analysis by setting $\alpha = 1$. For Case-II-b, we provide a different analysis. In Case II-b: we have $z \in \V(S_2)$ and $h(S_2) > h(S_2 - \{z\})$. Suppose $k=2$. Then, it is trivial that the sequence selected by Algorithm~\ref{alg:robustGreedy_tau_contiguous} (i.e., $S=(v_1) \oplus (z)$) yields an optimal solution. This is because removing element $z$ from $S$ gives $(v_1)$, which has the largest individual value among all the elements. Therefore, we assume $k>2$ throughout the rest of the proof. Recall that $\eta \triangleq \frac{h(S_2) - h(S_2 - \{z\})}{h(S_2)}$. We first state the following:
\begin{subequations}
\begin{align}
    & h(S_2) \geq  \frac{e^{\mu_1 \cdot \frac{k-2}{k-1}} - 1}{e^{\mu_1 \cdot \frac{k-2}{k-1}} - \mu_1(\eta + \mu_2 -1)} g_{\tau}(S^\ast(\V, k, \tau)) \label{eq:approximate_S_2_tau_1_large_k},\\
    &  \max \left\{ \frac{\mu_1(\eta + \mu_2 -1) (e^{\mu_1 \cdot \frac{k-2}{k-1}} - 1)}{e^{\mu_1 \cdot \frac{k-2}{k-1}} - \mu_1(\eta + \mu_2 -1)},  \frac{(1 - \eta) (e^{\mu_1 \cdot \frac{k-2}{k-1}} - 1)}{e^{\mu_1 \cdot \frac{k-2}{k-1}} - \mu_1(\eta + \mu_2 -1)} \right\} \nonumber \\
    & \geq \frac{\mu_1 \mu_2}{1 + \mu_2} \cdot \frac{e^{\mu_1 \cdot \frac{k-2}{k-1}} - 1}{ e^{\mu_1 \cdot \frac{k-2}{k-1}} -\frac{\mu_1 \mu_2}{1 + \mu_2}} \label{eq:approximate_max_two_terms_assumption_a_2}.
\end{align}
\end{subequations}
We will prove Eqs.~\eqref{eq:approximate_S_2_tau_1_large_k} and \eqref{eq:approximate_max_two_terms_assumption_a_2} later; for now, we assume that they both hold. 
Then, we can obtain the following bound:
\begin{equation}
   \label{eq:approximate_s_z_case_II_b_tau_1_large_k}
    \begin{aligned}
       h(S - \{z\})
       & \geq \max \{\mu_1(\eta + \mu_2 -1) \cdot h(S_2),   (1 - \eta) \cdot h(S_2)\} \\
       & \geq \max \left\{ \frac{\mu_1(\eta + \mu_2 -1) (e^{\mu_1 \cdot \frac{k-2}{k-1}} - 1)}{e^{\mu_1 \cdot \frac{k-2}{k-1}} - \mu_1(\eta + \mu_2 -1)},  \frac{(1 - \eta) (e^{\mu_1 \cdot \frac{k-2}{k-1}} - 1)}{e^{\mu_1 \cdot \frac{k-2}{k-1}} - \mu_1(\eta + \mu_2 -1)} \right\}
       g_{\tau}(S^\ast(\V, k, \tau)) \\
       & \geq \frac{\mu_1 \mu_2}{1 + \mu_2} \cdot \frac{e^{\mu_1 \cdot \frac{k-2}{k-1}} - 1}{ e^{\mu_1 \cdot \frac{k-2}{k-1}} -\frac{\mu_1 \mu_2}{1 + \mu_2}} \cdot g_{\tau}(S^\ast(\V, k, \tau)),
       \end{aligned}
   \end{equation}
where the three inequalities are from Eq.~\eqref{eq:approximate_s_z_case_II_b_tau_1_assumption_a} (with $\alpha = 1$), \eqref{eq:approximate_S_2_tau_1_large_k}, and \eqref{eq:approximate_max_two_terms_assumption_a_2}, respectively.

By combining the bounds in Eqs.~\eqref{eq:approximate_s_z_case_I_tau_1}, \eqref{eq:approximate_s_z_case_II_a_tau_1} (with $\alpha = 1$), and \eqref{eq:approximate_s_z_case_II_b_tau_1_large_k}, we obtain Eq.~\eqref{eq:approximate_tau1bound1}.  

Next, we show that Eq.~\eqref{eq:approximate_S_2_tau_1_large_k} holds. Let $v^1_2$ denote the first element of sequence $S_2$. Then, we have the following: 
\begin{equation*}
    \begin{aligned}
        h((v^1_2)) & \stackrel{\text{(a)}} \geq h((z)) \stackrel{\text{(b)}}  \geq \mu_1(\eta + \mu_2 -1) \cdot h(S_2),
    \end{aligned}
\end{equation*}
where (a) holds because element $v^1_2$ has the largest individual value among all elements in $S_2$ and (b) follows from Eq.~\eqref{eq:approximate_z2_tau_1_assumption_a}.Then, we can characterize the value of $h(S_2)$ as follows:
 \begin{equation*}
    \begin{aligned}
        h(S_2) & \geq \frac{e^{\mu_1 \cdot \frac{k-2}{k-1}} - 1}{e^{\mu_1 \cdot \frac{k-2}{k-1}} - \mu_1(\eta + \mu_2 -1)} h(S^\ast(\V \setminus \{v_1\}, k - \tau, 0)) \\
        & \geq  \frac{e^{\mu_1 \cdot \frac{k-2}{k-1}} - 1}{e^{\mu_1 \cdot \frac{k-2}{k-1}} - \mu_1(\eta + \mu_2 -1)} g_{\tau}(S^\ast(\V, k, \tau)),
    \end{aligned}
 \end{equation*}
\noindent where the inequalities are from Lemma~\ref{lemma:approximate_value_concentration} (where we replace $\V$, $S$, $S_1$, $k$, $k^\prime$, and $c$ with $\V \setminus \{v_1\}$, $S_2$, $(v^1_2)$, $k-1$, $k-2$, and $\mu_1(\eta + \mu_2 -1)$, respectively) and Lemma~\ref{lemma:Optimal_value_from_subset}, respectively.

Finally, we show that Eq.~\eqref{eq:approximate_max_two_terms_assumption_a_2} holds. We define two functions of $\eta$:
    $l_1(\eta) \triangleq \frac{\mu_1(\eta + \mu_2 -1) (e^{\mu_1 \cdot \frac{k-2}{k-1}} - 1)}{e^{\mu_1 \cdot \frac{k-2}{k-1}} - \mu_1(\eta + \mu_2 -1)}$ and $l_2(\eta) \triangleq \frac{(1 - \eta) (e^{\mu_1 \cdot \frac{k-2}{k-1}} - 1)}{e^{\mu_1 \cdot \frac{k-2}{k-1}} - \mu_1(\eta + \mu_2 -1)}$. 
It is easy to verify that for $k > 2$ and $\eta \in (0, 1]$, function $l_1(\eta)$ is monotonically increasing and function $l_2(\eta)$ is monotonically decreasing. Let $\eta^* = \frac{1 + \mu_1 - \mu_1 \mu_2}{1 + \mu_1 }$, and we have $l_1(\eta^*) = l_2(\eta^*) = \frac{\mu_1 \mu_2}{1 + \mu_2} \cdot \frac{e^{\mu_1 \cdot \frac{k-2}{k-1}} - 1}{ e^{\mu_1 \cdot \frac{k-2}{k-1}} -\frac{\mu_1 \mu_2}{1 + \mu_2}}$. We consider two cases for $\eta$: $\eta \in [\eta^*, 1]$ and $\eta \in (0, \eta^*]$. For $\eta \in [\eta^*, 1]$, we have $\max\{l_1(\eta), l_2(\eta)\} \geq l_1(\eta) \geq l_1(\eta^*)$ as $l_1(\eta)$ is monotonically increasing; for $\eta \in (0, \eta^*]$, we have $\max\{l_1(\eta), l_2(\eta)\} \geq l_2(\eta) \geq l_2(\eta^*) = l_1(\eta^*)$ as $l_2(\eta)$ is monotonically decreasing. Therefore, for $\eta \in (0, 1]$, we have $\max\{l_1(\eta), l_2(\eta)\} \geq l_1(\eta^*) = \frac{\mu_1 \mu_2}{1 + \mu_2} \cdot \frac{e^{\mu_1 \cdot \frac{k-2}{k-1}} - 1}{ e^{\mu_1 \cdot \frac{k-2}{k-1}} -\frac{\mu_1 \mu_2}{1 + \mu_2}}$. This gives Eq.~\eqref{eq:approximate_max_two_terms_assumption_a_2} and completes the proof.
\end{proof}

\subsection{Proof of Theorem \ref{theorem:robustGreedy_tau_contigous_case_b}}
\label{proof:robustGreedy_tau_contigous_case_b}
\begin{proof}
We use Lemma~\ref{lemma:Optimal_value_from_subset} presented in Appendix~\ref{subsec:preliminary_results} and Lemmas~\ref{lemma:approximate_classic_greedy} and \ref{lemma:approximate_value_concentration} presented in Appendix~\ref{subsec:approximate_preliminary_results} to prove the approximation ratios of Algorithm~\ref{alg:robustGreedy_tau_contiguous} in the case of $1 \leq \tau \leq k$, assuming the removal of $\tau$ contiguous elements. The proof follows a similar line of analysis as in the proof of Theorem~\ref{theorem:robustGreedy_tau_contigous_case_a}.

In Step 1 of Algorithm~\ref{alg:robustGreedy_tau_contiguous}, it is equivalent that sequence $S_1$ is selected by the SSG algorithm from set $\V$, and we have $|S_1| = \tau$. Similarly, in Step 2 of Algorithm~\ref{alg:robustGreedy_tau_contiguous}, it is also equivalent that sequence $S_2$ is selected by the SSG algorithm from set $\V \setminus \V(S_1)$, and we have $|S_2| = k-\tau$. Hence, the sequence selected by Algorithm~\ref{alg:robustGreedy_tau_contiguous} can be written as $S = S_1 \oplus S_2$. Recall that for any given sequence $S$, set $\Z_{\tau}(S)$ denotes the set of elements removed from sequence $S$ in the worst case (i.e., $\Z_{\tau}(S)$ is an optimal solution to Problem~\eqref{eq:Z_S}). We define $\Z^1_\tau(S) \triangleq \Z_\tau(S) \cap \V(S_1)$ and $\Z^2_\tau(S) \triangleq \Z_\tau(S) \cap \V(S_2)$ as the set of elements removed from subsequences $S_1$ and $S_2$, respectively.

We want to show the following two bounds, which establish the approximation ratios of Algorithm~\ref{alg:robustGreedy_tau_contiguous}: 
1) Suppose that function $h$ is forward-monotone, backward-monotone,  $\mu_1$-element-sequence-submodular, and $\mu_2$-sequence-submodular (Assumption~\ref{assump:forward_backward_mu1_mu2}), we have
\begin{equation}
    h(S - \Z_\tau(S)) \ge \frac{ a \mu_2 (e^b - 1)}{(a+1)e^b - a \mu_2}g_{\tau}(S^\ast(\V, k, \tau)),  \label{eq:approximate_tau_contigous_bound1}
\end{equation}
where $a = \mu_1 \cdot (1-1/e^{\mu_1})$ and $b = \mu_1 \cdot \frac{k-2\tau}{k-\tau}$;
2) Suppose that function $h$ is forward-monotone, $\alpha$-backward-monotone,  $\mu_1$-element-sequence-submodular, and $\mu_2$-sequence-submodular (Assumption~\ref{assump:forward_alpha_mu1_mu2}), we have
\begin{equation}
\label{eq:approximate_tau_contigous_bound2}
    h(S - \Z_\tau(S)) \ge \frac{\alpha^2 \mu_1 \mu_2 (e^{\mu_1} - 1)^2}{\mu_1 e^{\mu_1}(e^{\mu_1} - 1) + e^{2\mu_1}}g_{\tau}(S^\ast(\V, k, \tau)). 
\end{equation}

We begin with the proof of Eq.~\eqref{eq:approximate_tau_contigous_bound2} as some of the intermediate results will be used in the proof of Eq.~\eqref{eq:approximate_tau_contigous_bound1}.
We first present a lower bound on $h(S_2)$, which will be used throughout the proof:
\begin{equation}
\label{eq:approximate_h_S_2_main_tau_contigous}
        \begin{aligned}
         h(S_2) & \geq \alpha (1 - 1/e^{ \mu_1}) h(S^\ast(\V \setminus \V(S_1), k - \tau, 0)) \\ 
        & \geq \alpha (1 - 1/e^{ \mu_1}) g_{\tau}(S^\ast(\V, k, \tau)),
    \end{aligned}
    \end{equation}
where the first inequality is from Lemma~\ref{lemma:approximate_classic_greedy} (where we replace $\V$, $k$, and $i$ with $\V \setminus \V(S_1)$, $k-\tau$, and $k-\tau$, respectively) and the second inequality is from Lemma~\ref{lemma:Optimal_value_from_subset} (where we replace $\V^\prime$ with $\V(S_1)$).

The proof proceeds as follows. Elements in $\Z_\tau(S)$ will be removed from sequence $S$. These elements can be either fully or partially in $\V(S_1)$ (i.e., $\Z^2_\tau(S) = \emptyset$ or $\Z^2_\tau(S) \neq \emptyset$). Therefore, we consider two cases: (I) $\Z^2_\tau(S) = \emptyset$ and (II) $\Z^2_\tau(S) \neq \emptyset$. 

In Case I, we have $\Z^2_\tau(S) = \emptyset$, i.e., $\Z_\tau(S) = \V(S_1)$. Then, we have the following: 
\begin{equation}
\label{eq:approximate_s_z_case_I_tau_congtigous}
    \begin{aligned}
        h(S - \Z_\tau(S)) & = h(S_2)  \geq \alpha (1 - 1/e^{ \mu_1}) g_\tau(S^\ast(\V, k, \tau)),
    \end{aligned}
\end{equation}
where the inequality follows from Eq.~\eqref{eq:approximate_h_S_2_main_tau_contigous}. 

In Case II, we have $\Z^2_\tau(S) \neq \emptyset$.
Depending on the impact of removing elements in $\Z^2_\tau(S)$, we consider two subcases: (II-a) $h(S_2) \leq h(S_2 - \Z^2_\tau(S))$ and (II-b) $h(S_2) > h(S_2 - \Z^2_\tau(S))$.

In Case II-a: we have $h(S_2) \leq h(S_2 - \Z^2_\tau(S))$. In this case, the removal of elements in $\Z^2_\tau(S)$ does not reduce the overall value of the remaining sequence $S_2 - \Z^2_\tau(S)$. Then, we have
\begin{equation}
 \label{eq:approximate_s_z_case_II_a_tau_congtigous}  
    \begin{aligned}
        h(S - \Z_\tau(S)) & = h((S_1 - \Z^1_\tau(S)) \oplus (S_2 - \Z^2_\tau(S))) \\
        & \stackrel{\text{(a)}} \geq \alpha \cdot h(S_2 - \Z^2_\tau(S)) \\
        & \stackrel{\text{(b)}} \geq   \alpha \cdot  h(S_2) \\
        & \stackrel{\text{(c)}} \geq  \alpha^2 (1 - 1/e^{ \mu_1}) \cdot g_\tau(S^\ast(\V, k, \tau)),
    \end{aligned}
\end{equation}
where (a) is due to the $\alpha$-backward monotonicity of function $h$, (b) holds from the condition of this subcase, and (c) follows from Eq.~\eqref{eq:approximate_h_S_2_main_tau_contigous}.

In Case II-b: we have $h(S_2) > h(S_2 - \Z^2_\tau(S))$. Let $\eta \triangleq \frac{h(S_2) - h(S_2 - \Z^2_\tau(S))}{h(S_2)}$ denote the ratio of the loss of removing elements in $\Z^2_\tau(S)$ from sequence $S_2$ to the value of sequence $S_2$, and we have $\eta \in (0, 1]$ due to $h(S_2) > h(S_2 - \Z^2_\tau(S))$. We first state the following:
\begin{subequations}
\begin{align}
    & h(S - \Z_\tau(S)) \geq \max\{\mu_1 \alpha(\eta + \mu_2 -1) (1 - 1/e^{ \mu_1}) \cdot h(S_2),  \alpha (1 - \eta) \cdot h(S_2)\} \label{eq:approximate_s_z_case_II_b_tau_contigous_assumption_a},\\
    & \max\{\mu_1 \alpha(\eta + \mu_2 -1) (1 - 1/e^{ \mu_1}), \alpha(1-\eta)\} \geq \frac{\alpha \mu_1 \mu_2 (e^{ \mu_1} - 1)}{\mu_1 (e^{ \mu_1} - 1) + e^{\mu_1}}     \label{eq:approximate_max_two_terms_tau_contigous_assumption_a}.
\end{align}
\end{subequations}

We will prove Eqs.~\eqref{eq:approximate_s_z_case_II_b_tau_contigous_assumption_a} and \eqref{eq:approximate_max_two_terms_tau_contigous_assumption_a} later; for now, we assume that they both hold. Then, we can obtain the following bound:
\begin{equation}
 \label{eq:approximate_s_z_case_II_b_tau_congtigous_any_k}
    \begin{aligned}
        h(S - \Z_\tau(S))&\geq \max\{\mu_1 \alpha(\eta + \mu_2 -1) (1 - 1/e^{ \mu_1}) \cdot h(S_2),   \alpha (1 - \eta) \cdot h(S_2)\} \\
        & \geq \max\{\mu_1 \alpha(\eta + \mu_2 -1) (1 - 1/e^{ \mu_1}) ,   \alpha (1 - \eta) \} \cdot \alpha (1 - 1/e^{ \mu_1}) \cdot g_\tau(S^\ast(\V, k, \tau)) \\
        & \geq \frac{\alpha^2 \mu_1 \mu_2 (e^{\mu_1} - 1)^2}{\mu_1  e^{\mu_1}(e^{\mu_1} - 1) + e^{2\mu_1}} \cdot g_\tau(S^\ast(\V, k, \tau)),
    \end{aligned}
\end{equation}
where the three inequalities are from Eqs.~\eqref{eq:approximate_s_z_case_II_b_tau_contigous_assumption_a}, \eqref{eq:approximate_h_S_2_main_tau_contigous}, and \eqref{eq:approximate_max_two_terms_tau_contigous_assumption_a}, respectively.  

By combining the bounds in Eqs.~\eqref{eq:approximate_s_z_case_I_tau_congtigous}, \eqref{eq:approximate_s_z_case_II_a_tau_congtigous}, and \eqref{eq:approximate_s_z_case_II_b_tau_congtigous_any_k}, we obtain the second bound in Eq.~\eqref{eq:approximate_tau_contigous_bound2}.

Next, we show that Eq.~\eqref{eq:approximate_s_z_case_II_b_tau_contigous_assumption_a} holds. Recall that the $\tau$ elements in $\Z_\tau(S)$ form a contiguous subsequence of $S$. Then, the elements in $\Z^1_\tau(S)$ and $\Z^2_\tau(S)$ also form a contiguous subsequence of $S_1$ and $S_2$, respectively. We use $Z_1$ and $Z_2$ to denote the contiguous subsequence of elements in $\Z^1_\tau(S)$ and $\Z^2_\tau(S)$, respectively. We first rewrite $S_2$ as $S_2=S_2^1 \oplus Z_2 \oplus S_2^2$, where $S_2^1$ and $S_2^1$ denote the subsequences of $S_2$ before and after subsequence $Z_2$, respectively. Note that $S_2^1$ or $S_2^2$ could be an empty sequence, depending on the position of subsequence $Z_2$ in $S_2$. Then, we characterize $h(Z_2)$ in terms of $h(S_2)$:
 \begin{equation}
 \label{eq:intermediate_approximate_z2_tau_contigous_assumption_a}
    \begin{aligned}
        \eta \cdot h(S_2) & \stackrel{\text{(a)}} = h(S_2) - h(S_2 - \Z^2_\tau(S)) \\
        & = h(S_2^1 \oplus Z_2 \oplus S_2^2) - h(S_2^1 \oplus S_2^2) \\
        & = h(S_2^1) + h(Z_2|S_2^1) +  h(S_2^2|S_2^1 \oplus Z_2)  - h(S_2^1) -  h(S_2^2|S_2^1) \\
        & = h(Z_2|S_2^1) +  h(S_2^2|S_2^1 \oplus Z_2) -  h(S_2^2|S_2^1) \\
              & \stackrel{\text{(b)}} \leq h(Z_2|S_2^1) +  h(S_2^2|S_2^1 \oplus Z_2) -  \mu_2 h(S_2^2|S_2^1 \oplus Z_2) \\
        & = h(Z_2|S_2^1) +   ( 1 -  \mu_2) h(S_2^2|S_2^1 \oplus Z_2) \\
        & \stackrel{\text{(c)}} \leq h(Z_2|S_2^1) +   ( 1 -  \mu_2) h(S_2) \\
        & \stackrel{\text{(d)}} \leq h(Z_2)/\mu_1 + ( 1 -  \mu_2) h(S_2),
    \end{aligned}
 \end{equation}
  where (a) is from the definition of $\eta$, (b) is  due to function $h$ being $\mu_2$-sequence-submodular, (c) holds because $\mu_2 \leq 1$ and $h(S_2^2|S_2^1 \oplus Z_2) \leq h(S_2^2|S_2^1 \oplus Z_2) + h(S_2^1 \oplus Z_2) = h(S_2)$ (recall that $S_2 = S_2^1\oplus Z_2 \oplus S_2^2$), and (d) is  due to function $h$ being $\mu_1$-element-sequence-submodular.
From Eq.~\eqref{eq:intermediate_approximate_z2_tau_contigous_assumption_a}, we have the following:
\begin{equation}
    \label{eq:approximate_z2_tau_contigous_assumption_a}
    h(Z_2) \geq \mu_1(\eta + \mu_2 -1) h(S_2).
\end{equation}

To prove Eq.~\eqref{eq:approximate_s_z_case_II_b_tau_contigous_assumption_a}, we decompose it into two parts: (i) $h(S - \Z_\tau(S))  \ge \mu_1 \alpha(\eta + \mu_2 -1) (1 - 1/e^{ \mu_1}) \cdot h(S_2)$ and (ii) $h(S - \Z_\tau(S)) \ge \alpha (1 - \eta) \cdot h(S_2)$.

Let $\tau_2 \triangleq |\Z^2_\tau(S)|$. We use $S_1^{\tau_2}$ to denote the subsequence consisting of the first $\tau_2$ elements in $S_1$. Then, Part (i) can be shown through the following: 
\begin{equation}
\label{eq:approximate_s_z_case_II_b_tau_contigous_assumption_a_1}
    \begin{aligned}
        h(S - \Z_\tau(S)) & = h((S_1 - \Z^1_\tau(S)) \oplus (S_2 - \Z^2_\tau(S))) \\
        & \stackrel{\text{(a)}} \geq h(S_1 - \Z^1_\tau(S)) \\
        & = h(S_1^{\tau_2}) \\
        & \stackrel{\text{(b)}} \geq \alpha (1 - 1/e^{ \mu_1}) \cdot h(S^\ast(\V, \tau_2, 0)) \\
        & \stackrel{\text{(c)}} \geq \alpha  (1 - 1/e^{ \mu_1}) \cdot h(Z_2)\\
        & \stackrel{\text{(d)}} \geq \mu_1 \alpha(\eta + \mu_2 -1) (1 - 1/e^{ \mu_1}) \cdot h(S_2),
    \end{aligned}
 \end{equation}
\noindent where (a) is due to the forward monotonicity of function $h$, (b) is due to the greedy manner of selecting subsequence $S_1^{\tau_2}$ from set $\V$ and Lemma~\ref{lemma:approximate_classic_greedy} (where we replace both $k$ and $i$ with $\tau_2$), (c) holds because sequence $Z_2$ is a feasible solution to Problem~\eqref{eq:sequence_selection} for selecting a sequence of $\tau_2$ elements from $\V$, and (d) is from Eq.~\eqref{eq:approximate_z2_tau_contigous_assumption_a}.

Part (ii) can be shown through the following:
\begin{equation}
\label{eq:approximate_s_z_case_II_b_tau_contigous_assumption_a_2}
   h(S - \Z_\tau(S)) = h((S_1 - \Z^1_\tau(S)) \oplus (S_2 - \Z^2_\tau(S))) \stackrel{\text{(a)}} \geq \alpha \cdot h(S_2 - \Z^2_\tau(S)) \stackrel{\text{(b)}} =  \alpha (1-\eta) \cdot h(S_2), 
 \end{equation}
 where (a) is from the $\alpha$-backward monotonicity of function $h$ and (b) is from the definition of $\eta$.

Eq.~\eqref{eq:approximate_max_two_terms_tau_contigous_assumption_a} holds trivially for any $\eta \in (0, 1]$ by setting $\mu_1 \alpha(\eta + \mu_2 -1) (1 - 1/e^{ \mu_1})$ and $\alpha(1-\eta)$ to be equal, solving for $\eta$, and plugging it back. This completes the proof of the bound in Eq.~\eqref{eq:approximate_tau_contigous_bound2}.

Now, we prove the bound in Eq.~\eqref{eq:approximate_tau_contigous_bound1}. Note that the analysis so far applies to any function $h$ that is forward-monotone, $\alpha$-backward-monotone,  $\mu_1$-element-sequence-submodular, and $\mu_2$-sequence-submodular (Assumption~\ref{assump:forward_alpha_mu1_mu2}), and hence, it also applies to  Assumption~\ref{assump:forward_backward_mu1_mu2}, which is a special case of Assumption~\ref{assump:forward_alpha_mu1_mu2} with $\alpha = 1$. The bound in Eq.~\eqref{eq:approximate_tau_contigous_bound1} requires backward monotonicity (i.e., $\alpha = 1$), but it becomes better than the bound in Eq.~\eqref{eq:approximate_tau_contigous_bound2} when $k$ is large.
In the following analysis, we assume that function $h$ is forward-monotone, backward-monotone,  $\mu_1$-element-sequence-submodular, and $\mu_2$-sequence-submodular (Assumption~\ref{assump:forward_backward_mu1_mu2}).

The proof proceeds as follows. We borrow the analysis of Case-I, Case~II-a, and Eq.~\eqref{eq:approximate_s_z_case_II_b_tau_contigous_assumption_a} from the previous analysis by setting $\alpha = 1$. Then, for Case-II-b, we provide a different analysis. 

In Case II-b: we have $h(S_2) > h(S_2 - \Z^2_\tau(S))$. Let $\tau_1 \triangleq |\Z^1_\tau(S)|$ and recall that $\tau_2 \triangleq |\Z^2_\tau(S)|$. Then, we have $\tau=\tau_1+\tau_2$ and $k = |S| = |S_1| + |S_2| \ge \tau + \tau_2$. Therefore, we consider two cases: $k = \tau + \tau_2$ and $k > \tau + \tau_2$. 

Suppose $k = \tau + \tau_2$. Then, it implies $\Z^2_\tau(S) = \V(S_2)$, i.e., all the elements in $S_2$ are removed. This further implies that the elements in $\Z^1_\tau(S)$ are at the end of sequence $S_1$ (due to the contiguous assumption of elements in $\Z_\tau(S)$). Let $S_1^{\tau_2} \preceq S_1$ denote the subsequence consisting of the first $\tau_2$ elements in $S_1$. It is easy to see $S_1 - \Z^1_\tau(S) = S_1^{\tau_2}$. Then, we have the following:
\begin{equation}
\label{eq:approximate_h_S_1}
\begin{aligned}
         h(S - \Z_\tau(S)) & = h((S_1 - \Z^1_\tau(S)) \oplus (S_2 - \Z^2_\tau(S))) \\
         & = h(S_1 - \Z^1_\tau(S)) \\
         & = h(S_1^{\tau_2}) \\
         & \geq  (1 - 1/e^{\mu_1}) h(S^\ast(\V, \tau_2, 0)) \\ 
        & \geq  (1 - 1/e^{\mu_1}) g_{\tau}(S^\ast(\V, k, \tau)),
    \end{aligned}
\end{equation}
where the first inequality is from Lemma~\ref{lemma:approximate_classic_greedy} (where we replace both $k$ and $i$ with $\tau_2$ and set $\alpha = 1$) and the second inequality is due to $\tau_2=k-\tau$ and Lemma~\ref{lemma:Optimal_value_from_subset} (where $\V^\prime$ is an empty set).

Now, suppose $k>\tau + \tau_2$. Recall that $\eta \triangleq \frac{h(S_2) - h(S_2 - \Z^2_\tau(S))}{h(S_2)}$. Let $a \triangleq \mu_1 (1- 1/e^{\mu_1})$, $b \triangleq \mu_1 \cdot \frac{k-2\tau}{k-\tau}$, and $b^\prime \triangleq \mu_1 \cdot \frac{k-\tau - \tau_2}{k-\tau_2}$. We first state the following:
\begin{subequations}
\begin{align}
    & h(S_2) \geq \frac{e^{b^\prime} - 1}{e^{b^\prime} - a (\eta + \mu_2 -1)} g_\tau(S^\ast(\V, k, \tau))  \label{eq:approximate_S_2_large_k_tau_contigous_assumption_a},\\
    & \max \left\{\frac{a (\eta + \mu_2 -1)(e^{b^\prime} - 1)}{e^{b^\prime} - a (\eta + \mu_2 -1)},  \frac{(1 - \eta)(e^{b^\prime} - 1)}{e^{b^\prime} - a (\eta + \mu_2 -1)} \right\} \geq \frac{ a \mu_2 (e^{b^\prime} - 1)}{(a+1)e^{b^\prime} - a \mu_2} \label{eq:approximate_max_two_terms_tau_contigous_2_assumption_a}.
\end{align}
\end{subequations}

We will prove Eqs.~\eqref{eq:approximate_S_2_large_k_tau_contigous_assumption_a} and \eqref{eq:approximate_max_two_terms_tau_contigous_2_assumption_a} later; for now, we assume that they both hold. Then, we can obtain the following bound:
\begin{equation}
   \label{eq:approximate_s_z_case_II_b_tau_congtigous_large_k}
    \begin{aligned}
        h(S - \Z_\tau(S))&\geq \max\{\mu_1 (\eta + \mu_2 -1) (1 - 1/e^{ \mu_1}) \cdot h(S_2), (1 - \eta) \cdot h(S_2)\} \\
       & \geq \max \left\{\frac{a (\eta + \mu_2 -1)(e^{b^\prime} - 1)}{e^{b^\prime} - a (\eta + \mu_2 -1)},  \frac{(1 - \eta)(e^{b^\prime} - 1)}{e^{b^\prime} - a (\eta + \mu_2 -1)} \right\} g_\tau(S^\ast(\V, k, \tau)) \\
       & \geq \frac{ a \mu_2 (e^{b^\prime} - 1)}{(a+1)e^{b^\prime} - a \mu_2} g_\tau(S^\ast(\V, k, \tau)) \\
       & \geq \frac{ a \mu_2 (e^{b} - 1)}{(a+1)e^{b} - a \mu_2} g_\tau(S^\ast(\V, k, \tau)),
       \end{aligned}
   \end{equation}
where the first three inequalities are from Eq.~\eqref{eq:approximate_s_z_case_II_b_tau_contigous_assumption_a} ($\alpha = 1$) and Eqs.~\eqref{eq:approximate_S_2_large_k_tau_contigous_assumption_a} and \eqref{eq:approximate_max_two_terms_tau_contigous_2_assumption_a}, respectively; for the last inequality, we replace $\tau_2$ in $b^\prime$ with $\tau$ to obtain $b$, and the inequality holds due to $\tau_2 \leq \tau$ and that $\frac{ a \mu_2 (e^{b^\prime} - 1)}{(a+1)e^{b^\prime} - a \mu_2}$ is a decreasing function of $\tau_2$. 

By combining the bounds in Eqs.~\eqref{eq:approximate_s_z_case_I_tau_congtigous}, \eqref{eq:approximate_s_z_case_II_a_tau_congtigous} (with $\alpha = 1$), \eqref{eq:approximate_h_S_1}, and \eqref{eq:approximate_s_z_case_II_b_tau_congtigous_large_k}, we obtain Eq.~\eqref{eq:approximate_tau_contigous_bound1}.  
 
Next, we show that Eq.~\eqref{eq:approximate_S_2_large_k_tau_contigous_assumption_a} holds. Let $S_{2}^{\tau_2} \preceq S_2$ denote the subsequence consisting of the first $\tau_2$ elements of sequence $S_2$. Then, we have the following:
  \begin{equation}
    \begin{aligned}
    h(S_{2}^{\tau_2})& \stackrel{\text{(a)}} \geq (1 - 1/e^{\mu_1}) h(S^\ast(\V \setminus \V(S_1), \tau_2, 0)) \\
        & \stackrel{\text{(b)}} \geq (1 - 1/e^{\mu_1}) h(Z_2)\\
        & \stackrel{\text{(c)}} \geq \mu_1(\eta + \mu_2 -1) (1 - 1/e^{\mu_1}) h(S_2),
    \end{aligned}
 \end{equation}
where (a) is due to the greedy manner of selecting subsequence $S_{2}^{\tau_2}$ from set $\V \setminus \V(S_1)$ and Lemma~\ref{lemma:approximate_classic_greedy} (where we set $\alpha = 1$ and replace $\V$, $k$, and $i$ with $\V \setminus \V(S_1)$, $\tau_2$, and $\tau_2$, respectively), (b) holds because sequence $Z_2$ is a feasible solution to Problem \eqref{eq:sequence_selection} for selecting a sequence of $\tau_2$ elements from $\V = \V \setminus \V(S_1)$, and (c) is from Eq.~\eqref{eq:approximate_z2_tau_contigous_assumption_a}. 

Therefore, we can characterize the value of $h(S_2)$ as follows (recall that $b^\prime \triangleq \mu_1 \cdot \frac{k-\tau - \tau_2}{k-\tau_2}$):
 \begin{equation}
    \begin{aligned}
        h(S_2) & \geq  \frac{e^{b^\prime} - 1}{e^{b^\prime} - a(\eta + \mu_2 -1)} h(S^\ast(\V \setminus \V(S_1), k - \tau, 0)) \\
        & \geq \frac{e^{b^\prime} - 1}{e^{b^\prime} - a(\eta + \mu_2 -1)} g_\tau(S^\ast(\V, k, \tau)),
    \end{aligned}
 \end{equation}
where the inequalities are from Lemma~\ref{lemma:approximate_value_concentration} (where we replace $\V$, $S$, $S_1$, $k$, $k^\prime$, and $c$ with $\V \setminus \V(S_1)$, $S_2$, $S_{2}^{\tau_2}$, $k-\tau$, $k-\tau-\tau_2$, and $a(\eta + \mu_2 -1)$, respectively) and Lemma~\ref{lemma:Optimal_value_from_subset}, respectively.

Finally, we show that Eq.~\eqref{eq:approximate_max_two_terms_tau_contigous_2_assumption_a} holds. We define two functions of $\eta$:
\begin{equation*}
l_1(\eta) = \frac{a \cdot (\eta + \mu_2 -1)(e^{b^\prime} - 1)}{e^{b^\prime} - a \cdot (\eta + \mu_2 -1)} ~\text{and}~ l_2(\eta) = \frac{(1 - \eta)(e^{b^\prime} - 1)}{e^{b^\prime} - a \cdot (\eta + \mu_2 -1)}.
\end{equation*}
It can be verified that for $k > \tau + \tau_2$ and $\eta \in (0, 1]$, function $l_1(\eta)$ is monotonically increasing and function $l_2(\eta)$ is monotonically decreasing. 
Also, consider $\eta^* \triangleq \frac{1 + a - a \cdot \mu_2}{a+1}$. We have $l_1(\eta^*) = l_2(\eta^*) = \frac{ a \mu_2 (e^{b^\prime} - 1)}{(a+1)e^{b^\prime} - a \mu_2}$.
We consider two cases for $\eta$: $\eta \in [\eta^*, 1]$ and $\eta \in (0, \eta^*]$. For $\eta \in [\eta^*, 1]$, we have $\max\{l_1(\eta), l_2(\eta)\} \geq l_1(\eta) \geq l_1(\eta^*)$ as $l_1(\eta)$ is monotonically increasing; for $\eta \in (0, \eta^*]$, we have $\max\{l_1(\eta), l_2(\eta)\} \geq l_2(\eta) \geq l_2(\eta^*) = l_1(\eta^*)$ as $l_2(\eta)$ is monotonically decreasing. Therefore, for $\eta \in (0, 1]$, we have $\max\{l_1(\eta), l_2(\eta)\} \geq l_1(\eta^*) = \frac{ a \mu_2 (e^{b^\prime} - 1)}{(a+1)e^{b^\prime} - a \mu_2}$. This gives Eq.~\eqref{eq:approximate_max_two_terms_tau_contigous_2_assumption_a} and completes the proof of Eq.~\eqref{eq:approximate_tau_contigous_bound1}.
\end{proof}

\subsection{Proof of Theorem \ref{theorem:robustGreedy_tau_general_case_b}}
\label{proof:robustGreedy_tau_general_case_b}
\begin{proof}
Suppose that function $h$ is forward-monotone, $\alpha$-backward-monotone, $\mu_1$-element-sequence-submodular,  and $\mu_3$-general-sequence-submodular (Assumption~\ref{assump:forward_alpha_mu1_mu3}). We use Lemma~\ref{lemma:Optimal_value_from_subset} presented in Appendix~\ref{subsec:preliminary_results} and Lemma~\ref{lemma:approximate_classic_greedy} presented in Appendix~\ref{subsec:approximate_preliminary_results} to prove that Algorithm \ref{alg:robustGreedy_tau_general} achieves an approximation ratio of $\frac{\alpha^2 \mu_1 \mu_3 (e^{\mu_1} - 1)}{(\mu_1  + \alpha   \tau)e^{\mu_1}}$ in the case of $1 \leq \tau \leq k$, assuming the removal of an arbitrary subset of $\tau$ selected elements, which are not necessarily contiguous.

In Step 1 of Algorithm~\ref{alg:robustGreedy_tau_general}, sequence $S_1$ is selected by choosing $\tau$ elements with the highest individual values in a greedy manner, and we have $|S_1| = \tau$. In Step 2 of Algorithm~\ref{alg:robustGreedy_tau_contiguous}, it is equivalent that sequence $S_2$ is selected by the SSG algorithm from set $\V \setminus \V(S_1)$, and we have $|S_2| = k-\tau$. Hence, the sequence selected by Algorithm~\ref{alg:robustGreedy_tau_contiguous} can be written as $S = S_1 \oplus S_2$. Recall that for any given sequence $S$, set $\Z_{\tau}(S)$ denotes the set of elements removed from sequence $S$ in the worst case (i.e., $\Z_{\tau}(S)$ is an optimal solution to Problem~\eqref{eq:Z_S}). We define $\Z^1_\tau(S) \triangleq \Z_\tau(S) \cap \V(S_1)$ and $\Z^2_\tau(S) \triangleq \Z_\tau(S) \cap \V(S_2)$ as the set of elements removed from subsequences $S_1$ and $S_2$, respectively.

The proof of Theorem~\ref{theorem:robustGreedy_tau_general_case_b} follows a similar line of analysis as in the proof of Theorem~\ref{theorem:robustGreedy_tau_contigous_case_b}. Specifically, we will also consider three cases: (I) $\Z^2_\tau(S) = \emptyset$, (II-a) $\Z^2_\tau(S) \neq \emptyset$ and $h(S_2) \leq h(S_2 - \Z^2_\tau(S))$, and (II-b) $\Z^2_\tau(S) \neq \emptyset$ and $h(S_2) > h(S_2 - \Z^2_\tau(S))$. The proofs of Case I and Case II-a are identical to those in Theorem~\ref{theorem:robustGreedy_tau_contigous_case_b}. Therefore, we focus on Case~II-b, which requires a different proof strategy.

We want to show the following bound that establishes the approximation ratio of Algorithm~\ref{alg:robustGreedy_tau_general}:
\begin{equation}
\label{eq:approximate_tau_general_bound1}
    h(S - \Z_\tau(S)) \ge \frac{\alpha^2 \mu_1 \mu_3 (e^{\mu_1} - 1)}{(\mu_1  + \alpha   \tau)e^{\mu_1}} g_{\tau}(S^\ast(\V, k, \tau)).
\end{equation}

To begin with, we present a lower bound on $h(S_2)$, which will be used throughout the proof:
\begin{equation}
\label{eq:approximate_h_S_2_main_tau_general}
        \begin{aligned}
         h(S_2) & \geq \alpha (1 - 1/e^{\mu_1}) h(S^\ast(\V \setminus \V(S_1), k - \tau, 0)) \\ 
        & \geq \alpha (1 - 1/e^{\mu_1}) g_{\tau}(S^\ast(\V, k, \tau)),
    \end{aligned}
    \end{equation}
where the first inequality is from Lemma~\ref{lemma:approximate_classic_greedy} (where we replace $\V$, $k$, and $i$ with $\V \setminus \V(S_1)$, $k-\tau$, and $k-\tau$, respectively) and the second inequality is from Lemma~\ref{lemma:Optimal_value_from_subset} (where we replace $\V^\prime$ with $\V(S_1)$).

We now focus on Case II-b, where we have $\Z^2_\tau(S) \neq \emptyset$ and $h(S_2) > h(S_2 - \Z^2_\tau(S))$. Let $\eta \triangleq \frac{h(S_2) - h(S_2 - \Z^2_\tau(S))}{h(S_2)}$ denote the ratio of the loss of removing elements in $\Z^2_\tau(S)$ from sequence $S_2$ to the value of sequence $S_2$, and we have $\eta \in (0, 1]$. We first state the following:
\begin{subequations}
\begin{align}
    & h(S - \Z_\tau(S)) \geq \max\{\frac{\mu_1(\eta + \mu_3 -1)}{\tau} \cdot h(S_2),  \alpha (1 - \eta) \cdot h(S_2)\}, \label{eq:approximate_s_z_case_II_b_tau_general_assumption_a}\\
    & \max \left \{\frac{\mu_1(\eta + \mu_3 -1)}{\tau} , \alpha (1 - \eta) \right\} \geq  \frac{\alpha \mu_1 \mu_3}{\mu_1 + \tau \alpha}.     \label{eq:approximate_max_two_terms_tau_general_assumption_a}
\end{align}
\end{subequations}

We will prove Eqs.~\eqref{eq:approximate_s_z_case_II_b_tau_general_assumption_a} and \eqref{eq:approximate_max_two_terms_tau_general_assumption_a} later; for now, we assume that they both hold. 
Then, we can obtain the following bound:
\begin{equation}
 \label{eq:approximate_s_z_any_k_assumption_a_general_tau}
    \begin{aligned}
        h(S - \Z_\tau(S)) 
        & \geq \max \left \{\frac{\mu_1(\eta + \mu_3 -1)}{\tau} \cdot h(S_2),  \alpha (1 - \eta) \cdot h(S_2) \right \} \\
        & \geq \max \left \{\frac{\mu_1(\eta + \mu_3 -1)}{\tau} , \alpha (1 - \eta) \right\} \cdot \alpha (1 - 1/e^{\mu_1}) g_{\tau}(S^\ast(\V, k, \tau)) \\
        & \geq \frac{\alpha^2 \mu_1 \mu_3 (e^{\mu_1} - 1)}{(\mu_1  + \alpha   \tau)e^{\mu_1}} g_{\tau}(S^\ast(\V, k, \tau)),
    \end{aligned}
\end{equation}
where the three inequalities are from Eqs.~\eqref{eq:approximate_s_z_case_II_b_tau_general_assumption_a}, \eqref{eq:approximate_h_S_2_main_tau_general}, and \eqref{eq:approximate_max_two_terms_tau_general_assumption_a}, respectively.

Now, we show that Eqs.~\eqref{eq:approximate_s_z_case_II_b_tau_general_assumption_a} and \eqref{eq:approximate_max_two_terms_tau_general_assumption_a} hold. We start by characterizing the value of elements in $\Z^2_\tau(S)$. Let $\tau_2 \triangleq |\Z^2_\tau(S)|$, and let the elements in $\Z^2_\tau(S)$ be denoted by $z_1, z_2, \dots, z_{\tau_2}$ according to their order in sequence $S_2$. Then, we can rewrite $S_2$ as  $S_2 = S_2^1 \oplus (z_1) \oplus S_2^2 \oplus (z_2) \oplus \dots \oplus S_2^{\tau_2 + 1}$, where $S_2^i$ is the subsequence between elements $z_{i-1}$ and $z_i$, for $i = 1,2, \dots, \tau_2 + 1$, and both $z_0$ and $z_{\tau+1}$ are an empty sequence. Note that subsequence $S_2^i$ could be an empty sequence, for $i = 1, 2, \dots, \tau_2+1$. We characterize the value of elements in $\Z^2_\tau(S)$ in the following:
 \begin{equation}
 \label{eq:intermediate_approximate_z2_general_tau}
    \begin{aligned}
        \eta \cdot h(S_2)  \stackrel{\text{(a)}}=& h(S_2) - h(S_2 - \Z^2_\tau(S)) \\
        =& h(S_2^1) + h((z_1)|S_2^1) +  h(S_2^2|S_2^1 \oplus (z_1)) \\
        & + \dots + h(S_2^{\tau_2+1}|S_2^1 \oplus (z_1) \oplus \dots \oplus S_2^{\tau_2} \oplus (z_{\tau_2}))  \\
        & -  h(S_2^1) - h(S_2^2|S_2^1) - \dots   - h(S_2^{\tau_2+1}|S_2^1 \oplus \dots \oplus S_2^{\tau_2}) \\
        =& \sum_{i =1}^{\tau_2} h((z_i)|S_2^1 \oplus (z_1) \oplus \dots \oplus (z_{i-1}) \oplus S_2^{i})   + h(S_2^2|S_2^1 \oplus (z_1)) - h(S_2^2|S_2^1) \\
        & + \dots + h(S_2^{\tau_2+1}|S_2^1 \oplus (z_1) \oplus \dots \oplus S_2^{\tau_2} \oplus (z_{\tau_2}))  - 
        h(S_2^{\tau_2+1}|S_2^1 \oplus \dots \oplus S_2^{\tau_2}) \\
        \stackrel{\text{(b)}} \leq & \sum_{i =1}^{\tau_2} h((z_i)|S_2^1 \oplus (z_1) \oplus \dots \oplus (z_{i-1}) \oplus S_2^{i}) \\
        & + h(S_2^2|S_2^1 \oplus (z_1)) - \mu_3 \cdot h(S_2^2|S_2^1 \oplus (z_1)) \\
        & + \dots + h(S_2^{\tau_2+1}|S_2^1 \oplus (z_1) \oplus \dots \oplus S_2^{\tau_2} \oplus (z_{\tau_2})) \\
        &- 
        \mu_3 \cdot h(S_2^{\tau_2+1}|S_2^1 \oplus (z_1) \oplus \dots \oplus S_2^{\tau_2} \oplus (z_{\tau_2})) \\
        = & \sum_{i =1}^{\tau_2} h((z_i)|S_2^1 \oplus (z_1) \oplus \dots \oplus (z_{i-1}) \oplus S_2^{i}) + (1 - \mu_3) \cdot h(S_2^2|S_2^1 \oplus (z_1)) \\
        & + \dots + (1 - \mu_3) \cdot h(S_2^{\tau_2+1}|S_2^1 \oplus (z_1) \oplus \dots \oplus S_2^{\tau_2} \oplus (z_{\tau_2})) \\
         = & \sum_{i =1}^{\tau_2} h((z_i)|S_2^1 \oplus (z_1) \oplus \dots \oplus (z_{i-1}) \oplus S_2^{i}) \\
        & + (1 - \mu_3) ( h(S_2^2|S_2^1 \oplus (z_1))  + \dots +  h(S_2^{\tau_2+1}|S_2^1 \oplus (z_1) \oplus \dots \oplus S_2^{\tau_2} \oplus (z_{\tau_2}))) \\
        \stackrel{\text{(c)}} \leq & \sum_{i =1}^{\tau_2} h((z_i)|S_2^1 \oplus (z_1) \oplus \dots \oplus (z_{i-1}) \oplus S_2^{i}) + (1 - \mu_3) \cdot h(S_2)\\
        \stackrel{\text{(d)}} \leq& \sum_{i =1}^{\tau_2} h((z_i))/\mu_1 + (1 - \mu_3) \cdot h(S_2),
    \end{aligned}
 \end{equation}
 where (a) is from the definition of $\eta$, (b) is  due to function $h$ being $\mu_3$-general-sequence-submodular, (c) holds because $\mu_3 \leq 1$ and $S_2 = S_2^1 \oplus (z_1) \oplus S_2^2 \oplus (z_2) \oplus \dots \oplus S_2^{\tau_2 + 1}$, and (d) is due to function $h$ being $\mu_1$-element-sequence-submodular.
 From Eq.~\eqref{eq:intermediate_approximate_z2_general_tau}, we have the following: 
\begin{equation}
    \label{eq:approximate_z2_general_tau}
    \sum_{i =1}^{\tau_2} h((z_i)) \geq \mu_1(\eta + \mu_3 -1) h(S_2).
\end{equation}

To prove Eq.~\eqref{eq:approximate_s_z_case_II_b_tau_general_assumption_a}, we decompose it into two parts: (i) $h(S - \Z_\tau(S))  \ge \frac{\mu_1(\eta + \mu_3 -1)}{\tau} \cdot h(S_2)$ and (ii) $h(S - \Z_\tau(S)) \ge \alpha (1-\eta) \cdot h(S_2)$.

Let $v^\prime$ denote the first element in $S_1 - \Z^1_\tau(S)$. Then, Part (i) can be shown through the following:
 \begin{equation*}
        h(S - \Z_\tau(S)) \stackrel{\text{(a)}} \geq h((v^\prime)) 
        \stackrel{\text{(b)}} \geq \frac{1}{\tau_2}  \sum_{i =1}^{\tau_2} h((z_i))
        \stackrel{\text{(c)}} \geq \frac{\mu_1(\eta + \mu_3 -1)}{\tau_2} \cdot  h(S_2)
        \stackrel{\text{(d)}} \geq \frac{\mu_1(\eta + \mu_3 -1)}{\tau} \cdot h(S_2),
 \end{equation*}  
where (a) is due to the forward monotonicity of function $h$, (b) is due to $h(v^\prime) \geq h(z_i)$ for any $i = 1,2, \dots, \tau_2$, (c) is from Eq.~\eqref{eq:approximate_z2_general_tau}, and (d) is due to $\tau_2 \leq \tau$.

Part (ii) can be shown through the following:
\begin{equation*}
    \begin{aligned}
   h(S - \Z_\tau(S)) = h((S_1 - \Z^1_\tau(S)) \oplus (S_2 - \Z^2_\tau(S))) \stackrel{\text{(a)}} \geq \alpha \cdot h(S_2 - \Z^2_\tau(S)) \stackrel{\text{(b)}} = \alpha (1-\eta) \cdot h(S_2),
    \end{aligned}
 \end{equation*}
 where (a) is from the $\alpha$-backward monotonicity of function $h$ and (b) is from the definition of $\eta$.
 
Eq.~\eqref{eq:approximate_max_two_terms_tau_general_assumption_a} holds trivially for any $\eta \in (0, 1]$ by setting $\frac{\mu_1(\eta + \mu_3 -1)}{\tau} $ and $\alpha (1-\eta)$ to be equal, solving for $\eta$, and plugging it back. This completes the proof.
\end{proof}
\end{document}